\documentclass[prd,aps,twocolumn,a4paper,floatfix]{revtex4}
\def\p{\partial}
\def\Lie{{\cal L}}

\usepackage{graphicx,psfrag}
\usepackage{mathrsfs}
\usepackage{amsmath,amsfonts,amssymb,amsthm}

\newtheorem{lem}{Lemma} 
\newtheorem{prop}{Proposition} 
\newtheorem{cor}{Corollary} 

\begin{document}


\title{Hyperbolic formulations of General Relativity with 
Hamiltonian structure}

\author{David Hilditch}
\affiliation{Theoretical Physics Institute, University of 
Jena, 07743 Jena, Germany}

\author{Ronny Richter}
\affiliation{Mathematisches Institut, Universi\"at T\"ubingen, 
72076 T\"ubingen, Germany}

\begin{abstract}
With the aim of deriving symmetric hyperbolic free-evolution 
systems for GR that possess Hamiltonian structure and allow for 
the popular puncture gauge condition we analyze the hyperbolicity 
of Hamiltonian systems. We develop helpful tools which are 
applicable to either the first order in time, second order in 
space or the fully second order form of the equations of motion.
For toy models we find that the Hamiltonian structure can
simplify the proof of symmetric hyperbolicity. In GR we use a
special structure of the principal part to prove symmetric
hyperbolicity of a formulation that includes gauge conditions
which are very similar to the puncture gauge.
\end{abstract}

\maketitle

\tableofcontents


\section{Introduction}
\label{section:Introduction}


In the build up to the solution of the binary black hole 
problem in general relativity, a lot of effort was spent 
examining the Einstein equations as a system of
partial differential equations (PDEs). 
Continuous dependence of a solution on given data, or 
mathematical well-posedness, of the initial boundary value 
problem (IBVP) is an essential property both for meaningful physical
models and for numerical applications. When GR is
decomposed into and against 
spacelike hypersurfaces it becomes a system of ten coupled 
PDEs; six evolution equations and four constraint equations. 
The constraints must be satisfied at all times in the 
development of a spacetime. There are few statements about 
well-posedness for constrained
systems~\cite{andersson03c,PhysRevD.77.084007}, but since the 
constraints are compatible with the evolution equations we 
may consider free-evolution schemes, in which for the purposes 
of analysis the constraints are not assumed; suitable initial 
and boundary data must be given to guarantee their 
satisfaction. Well-posedness of the free IBVP is determined 
by the character of the principal part of the system and
the boundary conditions. In particular well-posedness of the 
IVP is guaranteed by strong hyperbolicity. For the IBVP the 
notion of symmetric hyperbolicity may be used to guarantee 
well-posedness~\cite{GKO:1995,Nagy:2004td}.

From the free-evolution PDEs point of view, the gauge freedom 
of GR has two aspects, the freedom to choose arbitrary equations of 
motion for the spacetime coordinates and the addition of the 
constraint equations to the evolution equations. Both aspects
alter the principal part and therefore potentially the well-posedness
of the IBVP. That the well-posedness of the IBVP depends upon 
the choice of coordinates is troublesome, since there may be gauge 
conditions of interest that do not give rise to a well-posed IBVP. 

A coordinate choice that has been applied successfully in numerical
GR is determined by the puncture gauge~\cite{Gundlach:2006tw,vanMeter:2006vi}.
Here we want to derive Hamiltonian formulations that allow for that 
gauge condition and possess a well-posed IBVP as a free-evolution 
system. Our methods are also applicable when other gauge conditions 
are desired.

The motivation is twofold. Firstly every Hamiltonian system has a 
conserved quantity. The existence of a conserved quantity is also 
part of the defining property of symmetric hyperbolic PDEs. Therefore
one may expect a connection between the concepts. The relationship 
between symmetric hyperbolicity and Hamiltonian structure has not 
been previously studied by numerical relativists. After analyzing 
several symmetric hyperbolic Hamiltonian systems we find that the two 
energies are in general not related. However in some special cases 
the Hamiltonian energy can be used to simplify the analysis of 
symmetric hyperbolicity.

The second reason to consider Hamiltonian systems are their
convenient properties. As discussed above Hamiltonian systems possess
a conserved energy. Moreover it is guaranteed that the time translation
map is symplectic. The literature contains numerical methods
that preserve one of those properties exactly~\cite{Hairer-2009,HaLW06},
which makes Hamiltonian formulations of GR potentially interesting for
numerical simulation. Previous studies of symplectic integrators in
numerical GR show that indeed their application can have advantages
over standard Runge-Kutta 
schemes~\cite{PhysRevD.73.024001,Richter:2008pr,Richter:2009ff}.

In section~\ref{section:Hyperbolicity} we describe the basic 
concepts of hyperbolicity for fully second order systems. We then 
introduce Hamiltonian systems in section~\ref{sec:Hamilton_systems}.
In section~\ref{sec:Tools} we present several new tools that are 
useful in the analysis of hyperbolicity. Next those tools are 
applied to two toy models in section~\ref{section:Maxwell_Pure_Gauge}.
We then turn our attention to the Einstein equations. 
Section~\ref{section:ADM_Gauge} contains our Hyperbolicity 
analysis for a large class of Hamiltonian formulations of GR 
with live gauges. Finally we conclude in 
section~\ref{section:Conclusion}.


\section{Hyperbolicity of second order systems}
\label{section:Hyperbolicity}


We start by introducing strong and symmetric hyperbolicity for 
first order in time, second order in space as well as fully 
second order systems. The former case was treated comprehensively 
in~\cite{Gundlach:2005ta}. For fully second order systems one 
can find material on strong hyperbolicity 
in~\cite{Lax,kreiss-2002-604}. Here we also discuss symmetric 
hyperbolicity of fully second order systems.


\subsection{Definitions}
\label{section:Definitions}


\paragraph*{Well-posedness and hyperbolicity:} Well-posedness
is the requirement that an initial (boundary) value problem
has a solution that is unique and depends continuously on the 
initial (and boundary) data. It can be 
shown~\cite{GKO:1995,Nagy:2004td} that the first order evolution 
system
\begin{equation}
\p_tu^\mu={A^{p\mu}}_{\nu}\p_pu^{\nu}+S^{\mu}\label{eqn:evolution}
\end{equation}
has a well-posed initial (boundary) value problem if it is
\emph{strongly (symmetric) hyperbolic}.

The system \eqref{eqn:evolution} is called
strongly hyperbolic if the \emph{principal symbol}, the
contraction with a spatial vector $s_i$ of the principal part 
${A^{p\mu}}_{\nu}s_p$, has real eigenvalues and a complete set 
of eigenvectors which depend continuously on $s_i$. The system 
is called symmetric hyperbolic if there exists a hermitian 
{\it symmetrizer} $H$ such that (suppressing some matrix 
indices)
\begin{equation}
\label{eq:candidate_eq}
HA^p=(HA^p)^\dagger.
\end{equation}
and where $H$ is positive definite. A hermitian $\bar H$ that 
satisfies~\eqref{eq:candidate_eq} but is not necessarily positive definite 
is called a \emph{candidate symmetrizer}. Symmetric hyperbolicity is a 
more strict condition than strong hyperbolicity 
\cite[appendix C]{Gundlach:2005ta}. It is equivalent to 
the existence of a conserved positive definite energy 
\begin{align}
E&=\int\epsilon\textrm{d}x,&
\epsilon&=u^\dagger H u
\end{align}
for the principal part of the system linearized around some 
background.

These definitions were used to define strong and symmetric 
hyperbolicity for first order in time, second order in space 
systems \cite{Gundlach:2005ta} of the form
\begin{subequations}
\label{eq:FOITSOIS_def}
\begin{align}
\p_t v &= A_1^i \partial_i v + A_1 v +A_2 w + a,\\
\p_t w &= B_1^{ij}\partial_i\partial_jv + B_1^i \partial_iv+B_1 v 
+B_2^i\partial_iw+B_2 w + b.
\end{align}
\end{subequations}
In what follows we also consider fully second order systems of the form
\begin{equation}
\p_t^2q={\mathcal A}^{ij}\p_i\p_jq+{\cal B}^i\p_i\p_tq+\mathcal S.
\label{eq:evolution_second}
\end{equation}
In fact the set of fully second order systems is a special case of 
the first order in time, second order in space systems. However in both
forms of the equations certain calculations are simplified.

The system \eqref{eq:FOITSOIS_def} is called strongly (symmetric) 
hyperbolic if there exists a reduction to first order that is strongly 
(symmetric) hyperbolic. We likewise call the 
system~\eqref{eq:evolution_second} strongly (symmetric) hyperbolic
if there exists a fully first order reduction that is strongly 
(symmetric) hyperbolic. 

With those definitions a system is strongly hyperbolic if and only if 
the \emph{principal symbol}, $P^s$, has a complete set of eigenvectors 
(with real eigenvalues) that depend continuously on a spatial 
vector $s_i$. Strong hyperbolicity is equivalent to the existence of a 
complete set of characteristic variables with real 
speeds~(\cite{Gundlach:2005ta,Lax, kreiss-2002-604} and 
appendix~\ref{App:Hyperbolicity}). For first order in time, second order 
in space systems the principal symbol is
\begin{align}
P_1^s=\left(
\begin{array}{cc}
A_1^i s_i & A_2\\
B_1^{ij}s_i s_j & B_2^i s_i
\end{array}
\right),
\end{align}
and for fully second order systems it becomes
\begin{align}
\label{eq:fully_2nd_P_symbol}
P_2^s=\left(
\begin{array}{cc}
0 & \mathbf{1}\\
\mathcal A^{ij}s_i s_j & \mathcal B^i s_i
\end{array}
\right).
\end{align}

The existence of a symmetric hyperbolic first order reduction
is equivalent to the existence of a conserved positive definite 
energy, $E=\int\epsilon dx$ (again details can be found in 
\cite{Gundlach:2005ta, Lax, kreiss-2002-604} and appendix 
\ref{App:Hyperbolicity} respectively). For first order in time, second 
order in space systems the energy density has the form
\begin{align}
\nonumber
\epsilon_1&=u_i^\dag
H_1^{ij}(v)
u_j\\
&=
\left(
\begin{array}{c}
\p_i v\\
w
\end{array}
\right)^\dag
\left(
\begin{array}{cc}
H_{11}^{ij} & H_{12}^i\\
H_{12}^{j\,\dag} & H_{22}
\end{array}
\right)
\left(
\begin{array}{c}
\p_j v\\
w
\end{array}
\right),
\end{align}
and for fully second order systems it is
\begin{align}
\nonumber
\epsilon_2&=
u_i^\dagger
H_2^{ij}(q)
u_j\\
&=
\left(
\begin{array}{c}
\p_iq\\
\p_tq
\end{array}
\right)^\dagger
\left(
\begin{array}{cc}
H_{11}^{ij} & H_{12}^i\\
H_{12}^{j\,\dag} & H_{22}
\end{array}
\right)
\left(
\begin{array}{c}
\p_jq\\
\p_tq
\end{array}
\right).
\end{align}
In both cases we denote $H^{ij}$ a \emph{symmetrizer}. Furthermore if there
is a matrix $\bar H^{ij}$ such that $u_i^\dag\bar H^{ij} u_j$ is a conserved
quantity, but not necessarily positive definite then we call $\bar H^{ij}$
a \emph{candidate symmetrizer}.

Since there is only one time coordinate, if a first order in time, second 
order in space system~\eqref{eq:FOITSOIS_def} is obtained as the reduction 
of a fully second order system~\eqref{eq:evolution_second}, then the two 
systems will have the same level of hyperbolicity. 

In App.~\ref{App:Hyperbolicity} we demonstrate the statements about
equivalence of the various flavors of hyperbolicity for fully second order
systems.

We denote the matrices
\begin{align}
\label{eq:def_PP_matrix}
A^p{}_i{}^j &=
\left(
\begin{array}{cc}
A_1^j\delta^p_i & A_2\delta^p{}_i\\
B_1^{pj} & B_2^{p}
\end{array}
\right),&
\mathfrak A^p{}_i{}^j &=
\left(
\begin{array}{cc}
0 & \delta^p{}_i\\
\mathcal A^{pj} & \mathcal B^{p}
\end{array}
\right)
\end{align}
the principal part matrices of the first order in time
second order in space and the fully second order system
respectively.


\subsection{Symmetric hyperbolicity of fully second order 
systems}\label{sec:Sym_Hyp_FSO_systems}


Since the structure of the matrices $\mathfrak A^p{}_i{}^j$ 
in equation~\eqref{eq:def_PP_matrix} is quite simple one can 
derive another general criterion for fully second order systems 
of the form \eqref{eq:evolution_second} to be symmetric 
hyperbolic:
\newcounter{repeatedLemma}
\setcounter{repeatedLemma}{\thelem}
\begin{lem}
Symmetric hyperbolicity of the fully second order 
system~\eqref{eq:evolution_second} is equivalent to the existence 
of {\it a second order symmetrizer and fluxes $(H_1,\phi^i,\phi^{ij}$)} 
satisfying
\begin{align}
\label{eq:fluxes_props_text}
\phi^{\dag i}&=\phi^{i},&\phi^{ij}&=\phi^{[ij]}=\phi^{ji\,\dag},
\end{align}
\begin{align}
\nonumber
s_is_js_k(\phi^i - \mathcal{B}^{\dag i}H_1)&\mathcal{A}^{jk}=\\
&=\mathcal{A}^{\dag jk}(\phi^i-H_1\mathcal{B}^{i})s_is_js_k,
\label{eq:2nd_Cons_first_text}
\end{align}
for every spatial vector $s_i$, and
\begin{align}
\label{eq:2nd_Positivity_text}
&H^{ij} =\\
\nonumber
&\;\;\left(\begin{array}{cc}
H_1\mathcal{A}^{ij}+\mathcal{B}^{(i\,\dag}H_1\mathcal{B}^{j)}-
\mathcal{B}^{(i\,\dag}\phi^{j)} + \phi^{ij\,\dag} & \phi^i 
- \mathcal{B}^{\dag i}H_1  \\  
\phi^j - H_1\mathcal{B}^{j} & H_1
\end{array}\right)
\end{align}
hermitian positive definite.
\label{lem:fully_2nd_sym_hyp}
\end{lem}

The lemma is proven in appendix \ref{App:Hyperbolicity}.
To simplify the criterion further we notice the following.
If $T^i{}_j$ is an invertible matrix then positivity
of \eqref{eq:2nd_Positivity_text} is equivalent to 
$T^i{}_jH^{jk}T^\dag_k{}^{j}$ being positive definite. We 
choose
\begin{align}
T^i{}_j =
\left(
\begin{array}{cc}
L^\dag \delta^i{}_j & L^\dag \mathcal B^i - L^\dag\phi^iLL^\dag\\
0 & L^\dag
\end{array}
\right),
\end{align}
with $H_1=L^{-\dag}L^{-1}$. This is possible because positivity of 
\eqref{eq:2nd_Positivity_text} implies positivity of $H_1$. The 
invertible matrix $L^{-\dag}$ can be chosen to be the Cholesky
decomposition of $H_1$.

It follows that positivity of \eqref{eq:2nd_Positivity_text} is
equivalent to positivity of
\begin{align}
L^{-1}\mathcal A^{ij}L + L^\dag\phi^i\mathcal B^j L - L^\dag 
\phi^iLL^\dag\phi^j L
+L^\dag\phi^{ij\,\dag}L.
\end{align}
If we redefine $\phi^i$ and $\phi^{ij}$ appropriately then we get
the following
\begin{cor}
The fully second order system~\eqref{eq:evolution_second} is
symmetric hyperbolic if and only if there exists an invertible matrix $L$
and fluxes $(\phi^i,\phi^{ij})$ satisfying
\begin{align}
\label{eq:fluxes_props_cor}
\phi^{\dag i}&=\phi^{i},&\phi^{ij}&=\phi^{[ij]}=\phi^{ji\,\dag}
\end{align}
such that
\begin{align}
\nonumber
s_is_js_k(\phi^i - &L^{\dag}\mathcal{B}^{i\,\dag}L^{-\dag})L^{-1}\mathcal{A}^{jk}L=\\
&=L^{\dag}\mathcal{A}^{jk\,\dag}L^{-\dag}(\phi^i-L^{-1}\mathcal{B}^{i}L)s_is_js_k,
\label{eq:2nd_Cons_first_cor}
\end{align}
for every spatial vector $s_i$, and the matrix
\begin{align}
L^{-1}\mathcal A^{ij}L + \phi^{(i}L^{-1}\mathcal B^{j)} L - \phi^{(i}\phi^{j)}
+\phi^{ij\,\dag}
\label{eq:2nd_Positivity_text_cor}
\end{align}
hermitian positive definite.
\label{cor:fully_2nd_sym_hyp}
\end{cor}

The conditions of corollary \ref{cor:fully_2nd_sym_hyp} simplify significantly
in the special case that $\mathcal{B}^i=0$. In that case symmetric
hyperbolicity is equivalent to the existence of an invertible
matrix $L$ and fluxes $\phi^{ij}=\phi^{[ij]}$ such that
\begin{align}
\label{eq:2nd_B_vanish_cor}
H_3{}^{ij}&=L^{-1} \mathcal{A}^{ij} L + \phi^{ij\,\dag}
\end{align}
is hermitian positive definite. We can choose $\phi^i=0$
in corollary \ref{cor:fully_2nd_sym_hyp} without affecting positivity,
because $\phi^{(i}\phi^{j)}$ is clearly positive semi-definite.


\section{Hamiltonian systems}\label{sec:Hamilton_systems}


In this article we aim to analyze the hyperbolicity of 
Hamiltonian systems. Here we give a short introduction to the basic notions.

\subsection{Variational principle}
In many cases the dynamics of a physical system can be described 
by a single functional, the action
\begin{align}
S = \int L(q,\dot q)dt,
\end{align}
with the \emph{Lagrangian} $L$, where $\dot q$ denotes the time
derivative of $q$.

According to Hamilton's principle classical systems behave such that
the action is minimal. This leads to the \emph{Euler-Lagrange} equations
\begin{align}
\label{eq:EL-eqn}
\frac{d}{dt}\frac{\partial L}{\partial\dot q}-\frac{\partial L}{\partial q} 
= 0.
\end{align}
The ADM equations \cite{Arnowitt:1962hi}
are of the form \eqref{eq:EL-eqn}, if the extrinsic curvature, $K$,
is interpreted as an abbreviation for the time derivative of the
3-metric, $\gamma$, using
\eqref{eq:definition_K}. They are composed of the
first order in time equations \eqref{eq:ADM_constraints}
and the second order equations \eqref{eq:dynamical_ADM_b}.
The former being the constraints and the latter the dynamical ADM equations.

As mentioned in the previous section we are only interested in
free evolution systems of the form \eqref{eq:FOITSOIS_def} or
\eqref{eq:evolution_second}.

For general relativity that means we consider systems whose solutions 
can be directly related to solutions of the ADM equations if the
constraints are satisfied, but that also possess constraint
violating solutions. 

The term in the Euler-Lagrange equations \eqref{eq:EL-eqn}
that contains the highest order time derivatives is
\begin{align}
\frac{d}{dt}\frac{\partial L}{\partial \dot q}
= \frac{\partial^2 L}{\partial \dot q\partial q}\,\dot q
+ \frac{\partial^2 L}{\partial \dot q\partial \dot q}\,\ddot q\,.
\end{align}
Hence, for a free evolution system 
the matrix
\begin{align}
\label{eq:Def_Minv}
M:=\frac{\partial^2 L}{\partial\dot q\partial\dot q}
\end{align}
is invertible (the equation of motion of every variable may 
be solved its second time derivative). For this type of system one 
can simplify the structure of the Euler-Lagrange equations by writing 
them in their \emph{Hamiltonian formulation}.

\subsection{Hamiltonian structure}
A Hamiltonian formulation is given by the specification of a Hamiltonian 
${\cal H}$ for the system, constructed from so-called 
\emph{canonical positions} and \emph{momenta} $(q;p)$. It can be shown 
that the Hamiltonian is a conserved quantity of the system, the total 
energy.

Here the Hamiltonian can be expressed as the integral 
over space of a Hamiltonian density 
\begin{align}
{\mathcal H}&=\int_{\Sigma_t} d^3x\,{\mathcal H}_D(q,\p_i q, p,\p_i p),
\end{align}
where the Hamiltonian density is a local function of the canonical
variables and their spatial derivatives.
One obtains the following \emph{canonical equations of motion}
\begin{subequations}
\label{eq:canon_eqn}
\begin{align}
\p_tq&= \frac{\p{\mathcal H}_D}{\p p}-\p_i
\frac{\p{\mathcal H}_D}{\p(\p_i p)},\\
\p_tp&=-\frac{\p{\mathcal H}_D}{\p q}+\p_i
\frac{\p{\mathcal H}_D}{\p(\p_i q)}.
\end{align}
\end{subequations}
A system that can be written in the form~\eqref{eq:canon_eqn} 
is said to have~\emph{Hamiltonian structure}.

If one treats the Hamiltonian formulation of a model as fundamental
then also free evolution equations with singular matrix $M^{-1}$ exist.
In some situations it may even be possible to reduce those systems
to a fully second order form. We consider only formulations in which 
one can use the standard Legendre transformation to convert 
between \eqref{eq:EL-eqn} and \eqref{eq:canon_eqn}.


\subsection{Canonical equations of motion}
\label{section:Observations}


For systems with Hamiltonian structure one may easily show that
a candidate symmetrizer exists. We consider a Hamiltonian 
formulation of some quasilinear field theory in which the 
functional form of the Hamiltonian density is given by 
\begin{align}
{\cal H}_D&=\frac12 u^\dag \bar H(q)u+S(q,\p_iq,p)\label{eqn:Hamiltonian}
\end{align}
with $u^\dag=(\p_i q^\dag, p^\dag)$, where
\begin{align}
\label{eq:Ham_Matrix}
\bar H^{ij}&=\left(\begin{array}{cc}
V^{ij} & F^{i\,\dag}      \\ F^j  & M^{-1}
\end{array}\right)
\end{align}
and $S(q,\p_iq,p)$ does not contribute to the principal part of 
the system. Without loss of generality we can assume that
\[
M^{-1}=M^{-\dag},\quad V^{ij} = V^{ji\,\dag}, \quad V^{ij}=V^{ji}.
\]
The ansatz \eqref{eqn:Hamiltonian} might seem very restrictive,
because e.g. the Hamiltonian density for the dynamical ADM equations
contains curvature terms, i.e. second order spatial
derivatives of the 3-metric. However using integration by parts,
those terms can be transformed to match with the ansatz
\eqref{eqn:Hamiltonian}. Integration by parts introduces boundary 
terms, but the equations of motion are unaffected.

The canonical equations of motion are
\begin{subequations}
\label{eq:Hamiltonian_PP}
\begin{align}
\p_tq &=M^{-1}p+F^i\p_iq+s_q,\\
\p_tp &=F^{i\,\dag}\p_ip + V^{ij}\p_j\p_iq+s_p,
\end{align}
\end{subequations}
where $s_p$ and $s_q$ denote terms that do not contribute to the
principal part of the system.

If we identify $q$ and $p$ with $v$ and $w$
respectively then we see that \eqref{eq:Hamiltonian_PP} is a
first order in time second order in space system of the form
\eqref{eq:FOITSOIS_def}.
A short calculation shows that
the matrix $\bar H$ from \eqref{eq:Ham_Matrix} is a
candidate symmetrizer for \eqref{eq:Hamiltonian_PP}, we denote it the
\emph{canonical candidate}. Thus, if $\bar H$ is positive definite,
like e.g. for the wave equation,
then the system \eqref{eq:Hamiltonian_PP} is automatically
symmetric hyperbolic.


\section{Tools to analyze hyperbolicity}\label{sec:Tools}


In this section we present several tools and important special
cases that are helpful in the analysis of hyperbolicity.


\subsection{Decoupled variables}\label{sec:Decoupled_variables}


The special case that simplifies our analysis the most is that 
of the principal part matrix having a special block structure.

An obvious special case is that of sets $(v_1,w_1)$ and $(v_2,w_2)$ 
of fields that decouple in the principal part of the system. It 
is easy to check that the system is strongly (symmetric) hyperbolic
if and only if the two subsystems are. Hence one may analyze two 
disjoint systems in isolation. The most important example in 
the context of general relativity is that of matter and spacetime 
variables.

Another important special case is that where there are pairs
$(v_1,w_1)$ and $(v_2,w_2)$ that in the principal part are coupled
through first order spatial derivative terms only and where the
couplings within each pair occur through zeroth and second order
spatial derivatives. The principal part matrix has the structure
\begin{align}
\label{eq:block_structure_PP}
A^p{}_i{}^{j} =\left(
\begin{array}{cccc}
0 & A^p_{12}{}_i{}^j & A^p_{13}{}_i & 0\\
A^p_{21}{}_i{}^j  & 0 & 0 & A^p_{24}{}_i{}\\
A^p_{31}{}^j & 0 & 0 & A^p_{34}\\
0 & A^p_{42}{}^j & A^p_{43} & 0
\end{array}
\right).
\end{align}
For many systems that are relevant 
for numerical relativity one can achieve this form, 
because one can distinguish between variables with an odd or even 
number of indices. Since only the metric can contract these indices 
the two groups of variables are usually not mixed.

It can be checked that if there is a symmetrizer $\tilde H$,
\begin{align}
\tilde H^{ij} =\left(
\begin{array}{cccc}
H_{11}^{ij} & H_{12}^{ij} & H_{13}^i & H_{14}^i\\
H_{12}^{ji\,\dag} & H_{22}^{ij} & H_{23}^i & H_{24}^i\\
H_{13}^{j\,\dag} & H_{23}^{j\,\dag} & H_{33} & H_{34}\\
H_{14}^{j\,\dag} & H_{24}^{j\,\dag} & H_{34} & H_{44}
\end{array}
\right)
\end{align}
of \eqref{eq:block_structure_PP} then
\begin{align}
\label{eq:partitioned_symmetrizer}
H^{ij} =\left(
\begin{array}{cccc}
H_{11}^{ij} & 0 & 0 & H_{14}^i\\
0 & H_{22}^{ij} & H_{23}^i & 0\\
0 & H_{23}^{j\,\dag} & H_{33} & 0\\
H_{14}^{j\,\dag} & 0 & 0 & H_{44}
\end{array}
\right)
\end{align}
is also a symmetrizer of that system.
The positivity of $H$ follows, because the
positivity of $\tilde H$ implies positivity of all its principal
minors (exchange rows and columns appropriately).
Furthermore it is straightforward to check that
$H$ defines a conserved quantity, because the principal part matrix
has the block structure \eqref{eq:block_structure_PP}.

Hence a system with a principal part matrix of the form
\eqref{eq:block_structure_PP} is symmetric hyperbolic if and only if
it admits a symmetrizer of the form \eqref{eq:partitioned_symmetrizer}.

For fully second order formulations the principal part matrix is
$\mathfrak A^p{}_i{}^j$ defined in \eqref{eq:def_PP_matrix}.
In that matrix the upper left block always vanishes and the upper
right block is the identity. Hence, those submatrices
automatically have the desired block structure \eqref{eq:block_structure_PP}
and the argument reduces to conditions on the block structure of
$\mathcal A^{ij}$ and $\mathcal B^i$. If they have the following
form
\begin{align}
\mathcal A^{ij} &=
\left(
\begin{array}{cc}
\mathcal A_{11}^{ij} & 0 \\
0 & \mathcal A_{22}^{ij}
\end{array}
\right),&
\mathcal B^{i} &=
\left(
\begin{array}{cc}
0 & \mathcal B_{12}^{i} \\
\mathcal B_{21}^{i} & 0
\end{array}
\right)
\end{align}
then one can assume that the symmetrizer has the block structure
\eqref{eq:partitioned_symmetrizer}.

These considerations about the block structure of symmetrizers
can be very helpful when one tries to derive a symmetrizer by
making an ansatz with certain parameters. A block structure like
\eqref{eq:partitioned_symmetrizer} halves the number of parameters 
in the ansatz, simplifying the calculations significantly.


\subsection{Searching for symmetrizers in a family of 
matrices -- the rank criterion}\label{sec:rank_criterion}


In this section we consider whether or not symmetrizers may 
be found in a set of matrices $\mathcal G$. We find a necessary 
condition for $\mathcal G$ to contain a symmetrizer. If 
$\mathcal G$ can be chosen big enough then one obtains a 
necessary condition for symmetric hyperbolicity. We refer to the 
condition as the rank criterion.

We find in applications that it is typically easier to check 
the rank criterion than the alternative, which is to take the set 
of candidate symmetrizers for a certain formulation and prove 
that the set contains no positive definite matrix.

Consider a family of formulations that is parametrized by 
$l$~\emph{formulation parameters}, $c\in\mathcal C^l\subset\mathbb R^l$.
For a formulation with parameters $c$ we denote the principal 
part matrix $A^p{}_i{}^j(c)\in\mathbb R^{k\times k}$ (this can be 
the fully second order or the first order in time, second order in 
space principal part matrix~\eqref{eq:def_PP_matrix}). Assume that
$A^p{}_i{}^j$ depends continuously on $c$.

For the criterion to be applicable we also need to know that in every 
neighborhood $U_c\subset\mathcal C^l$ of each point $c\in\mathcal C^l$ 
there are formulations that do not possess a symmetrizer in the set of 
matrices, $\mathcal G$. For our examples we find this situation,
e.g. when the set of strongly hyperbolic 
formulations inside $\mathcal C^l$ has lower dimensionality.

The idea is then to show that by continuity $\mathcal G$ can only contain
symmetrizers for $c$ if the space of candidates for $c$,
$\mathcal G_c\subset\mathcal G$, is bigger than the space of candidates
for a generic $\tilde c\in U_c$.

We assume that $\mathcal G$ is linearly parametrized, i.e. there 
exists a linear map $G^{ij}:\mathbb R^n\rightarrow \mathbb R^{k\times k}$,
that assigns hermitian $k\times k$-matrices to parameters, 
$g\in\mathbb R^n$ such that $\mathcal G = G^{ij}(\mathbb R^n)$.
According to \cite{Gundlach:2005ta} a matrix $G^{ij}(g)$ is a candidate 
symmetrizer of a formulation $c\in\mathcal C^l$, if and only if the 
following matrix is hermitian for every spatial vector $s$
\begin{align}
S_i G^{ij}(g) A^p{}_j{}^k(c) s_p S_k,
\end{align}
with
\begin{align}
S_i =
\left(
\begin{array}{cc}
s_i & 0\\
0 & 1
\end{array}
\right).
\end{align}
Since $G^{ij}$ is a linear map this defines linear equations for the 
parameters $g$ of the form
\begin{align}
B(c)g = 0,
\end{align}
where $B:\mathcal C^l\rightarrow\mathbb R^{m\times n}$ is a continuous 
map, because we assume that $A^p{}_i{}^j$ is continuous.

In appendix \ref{app:rank_criterion} we show the following. If for 
fixed $N\in\mathbb N$ there is in every neighborhood of $c$ a 
formulation $\tilde c$ with $\mbox{rank}(B(\tilde c))=N$ such that 
$\mathcal G$ does not contain symmetrizers for $\tilde c$ and 
$\mbox{rank}(B(c))\geq N$ then there is no symmetrizer for $c$ 
in $\mathcal G$.

If $\mathcal G$ is the set of all hermitian $k\times k$-matrices 
and the rank of $B$ is continuous at $c$ this implies that $c$ 
is not symmetric hyperbolic. Hence,
it is not necessary to discuss positivity 
of matrices, one only needs to calculate the rank of $B$.


\subsection{Symmetric hyperbolicity of special Hamiltonian 
systems}


In section~\ref{sec:Sym_Hyp_FSO_systems} we derived general 
criteria for a fully second order system to be symmetric 
hyperbolic. We now discuss simplifications that can be 
achieved when the system additionally has a reduction to 
first order in time, second order in space with Hamiltonian 
structure in the principal part.

Assume that there exist matrices
$M=M^\dag$, $F^i$ and $V^{ij}=V^{(ij)\,\dag}$ such that
\begin{align}
\p_t^2q&=M^{-1}[V^{ij}-F^{\dagger(i}MF^{j)}]\p_i\p_jq\nonumber\\
&+[F^iM^{-1}+M^{-1}F^{\dag i}]M\p_i\p_tq+\mathcal S,
\label{eqn:Ham_Second}
\end{align}
where $\mathcal S$ denotes terms that do not contribute to the 
principal part.

The underlying Hamiltonian structure guarantees the
system a canonical candidate symmetrizer
\begin{align}
\label{eq:canon_cand_fully_2nd}
\bar{\mathcal H}^{ij} &= M\left(
\begin{array}{cc}
\mathcal A^{ij} & 0\\
0 & \mathbf{1}
\end{array}
\right),
\end{align}
the matrices $M\mathcal B^i$ and $M\mathcal A^{ij}$ are hermitian
for every $i,j$:
\begin{align}
\label{eq:def_mathfrak_H}
\nonumber
\mathfrak H_2^i &= M\mathcal B^{i} = M F^i + F^{i\,\dag} M,\\
\mathfrak H_3^{ij} &= M\mathcal A^{ij} = V^{ij} - F^{(i\,\dag} M F^{j)}.
\end{align}

However one needs to be careful. Although the canonical candidate
\eqref{eq:canon_cand_fully_2nd} is block diagonal one cannot assume
that there exists a positive definite symmetrizer with this block 
structure. However if $\mathcal B^i$ vanishes and the system is 
symmetric hyperbolic then it is easy to check that there exists a 
block diagonal symmetrizer (this was discussed previously
for general fully second order systems).

There are several special cases where the existence of the
canonical candidate \eqref{eq:canon_cand_fully_2nd} simplifies
the construction of a positive symmetrizer:

\emph{1. If $M$ and $\mathfrak H_3^{ij}$ are positive definite, 
the system is automatically symmetric hyperbolic.}

The proof is trivial, because the canonical candidate is positive 
definite.\qed

\emph{2. If $\mathcal A^{ij}$, $\mathcal B^i$ and $M$ commute for 
every $i,j$, then the system is symmetric hyperbolic if 
$\mathcal A^{ij}$ has only positive eigenvalues.}

$M$ is hermitian and therefore diagonalizable. The
matrices $\mathcal A^{ij}$, $\mathcal B^i$ and $M$ commute,
therefore each $\mathcal A^{ij}$ and $\mathcal B^i$ is diagonal
in the basis where $M$ is diagonal. In this basis it is then 
obvious that $M\bar{\mathcal H}$ is a candidate and if 
$\mathcal A^{ij}$ has only positive eigenvalues then
$M\bar{\mathcal H}$ is also positive definite.
\qed

\emph{3. The system~(\ref{eqn:Ham_Second}) is symmetric hyperbolic 
if there exists a matrix $\mathcal S$ such that the two matrices
\begin{align}
\mathcal SM,\qquad \mathcal SM\mathcal A^{ij}
\label{eqn:Ham_Second_cond_1}
\end{align}
are hermitian positive definite and
\begin{align}
\mathcal S M\mathcal B^i\label{eqn:Ham_Second_cond_2}
\end{align}
is symmetric. The symmetrizer is then $\mathcal S\bar{\mathcal H}$.}

Positivity of $\mathcal S\bar{\mathcal H}$ is obvious because of conditions
\eqref{eqn:Ham_Second_cond_1} and a straightforward calculation shows
that $\mathcal S\bar{\mathcal H}$ is also a candidate symmetrizer because
of \eqref{eqn:Ham_Second_cond_2}.
\qed


\section{Toy problems}
\label{section:Maxwell_Pure_Gauge}


We now use the criteria of the previous sections to analyze 
several toy problems. Our aim is to demonstrate the methods 
without considering equations with a complicated structure.


\subsection{Electromagnetism}
\label{section:Maxwell}


The structure of the Maxwell equations is very similar to that of 
the Einstein equations. In particular they share a very similar 
gauge sector. To expand the vacuum Maxwell equations in flat-space 
\begin{subequations}
\begin{align}
\label{eq:Maxwell_a}
\p_iB^i=0,  & \qquad  \p_iE^i=0,\\
\label{eq:Maxwell_b}
 \p_tB_i=-(\p\times E)_i , & \qquad  \p_tE_i= (\p\times B)_i.
\end{align}
\end{subequations}
with their largest gauge freedom we introduce the standard 
scalar and vector potentials $\phi$ and $A^i$, 
\begin{align}
\label{eq:def_E_B}
B_i=(\p\times A)_i, & \qquad  E_i=-\p_tA_i-\p_i\phi.
\end{align}
The gauge freedom is in the time derivative of $\phi$ and the 
divergence of $A_i$.
\paragraph*{Hamiltonian structure:}
In the canonical variables
$(A^i;\pi_i=-E_i)$ the dynamical part \eqref{eq:Maxwell_a} of the
system has Hamiltonian density
\begin{subequations}
\begin{align}
\nonumber
\mathcal{H}_M&=\frac{1}{2}\left[\pi_i\pi^i-(\p_iA^i)(\p_jA^j)
+(\p_iA^j)(\p^iA_j)\right]\\
&\qquad +\phi C_m,
\end{align}
\end{subequations}
and in analogy with GR, the ``momentum" constraint is
\begin{align} 
C_m\equiv\p_i\pi^i&= 0.
\end{align}
Here $\phi$ is considered a given field. Following the 
approach of \cite{Brown:2008cca} with the Maxwell equations, we promote 
$\phi$ to the status of an evolved field by introducing a canonical 
momentum $\pi$. In flat-space the dynamical vacuum Maxwell equations
with live gauge conditions can 
be expanded and written with Hamiltonian density
\begin{subequations}
\begin{align}
\mathcal{H}_D&=\mathcal{H}_M + \mathcal{H}_G,\\
\mathcal{H}_G&= (c_1-1)\phi C_m -\frac{1}{2} c_2 \pi^2- c_3 \pi \p_i A^i.
\label{eq:Maxwell_Ham}
\end{align}
\end{subequations}
where $\mathcal{H}_G$ denotes the Hamiltonian for the gauge
sector of the theory. The canonical variables are
$(A^i,\phi;\pi_i,\pi)$. In principle the 
$c_i$ are arbitrary given scalar functions on the spacetime, 
however from the point of view of hyperbolicity analysis we
will treat them as constants. Since we are in flat-space we 
do not denote the densitization of the variables. It is 
straightforward to generalize the system to curved space. The
equations of motion are
\begin{subequations}
\label{eq:Maxwell_eq_of_motion}
\begin{align}
\label{eq:Maxwell_eq_of_motion_2}
\p_tA^i &=   \pi^i - c_1 \p^i\phi,\\
\label{eq:Maxwell_eq_of_motion_1}
\p_t\phi &=  - c_2 \pi - c_3 \p_iA^i,\\
\label{eq:Maxwell_eq_of_motion_4}
\p_t\pi_i &= \p^j\p_jA_i - \p_i\p^jA_j- c_3 \p_i\pi,\\
\label{eq:Maxwell_eq_of_motion_3}
\p_t\pi &=   - c_1 \p_i\pi^i.
\end{align}
\end{subequations}
Note that in~\eqref{eq:Maxwell_Ham} the appearance of $\phi C_m$ 
results in the addition of the momentum constraint to the evolution
equations and to a rescaling of $\phi$ (we have
\eqref{eq:Maxwell_eq_of_motion_2} instead of \eqref{eq:def_E_B}
with $\pi_i=-E_i$). The constraint addition may alter the hyperbolicity 
of the system. Freezing $\phi$, identifying $\pi = \Gamma-\p_iA^i$ and 
choosing the remaining constants appropriately reduces the system 
to the KWB formulation of electromagnetism \cite{Knapp:2002fm}.

To see that the canonical momentum of $\phi$ should obey the
constraint $\pi=0$, note that we may write
\begin{subequations}
\begin{align}
\p_t\pi &=- c_1 C_m,\\
\p_tC_m &=- c_3 \p_i\p^i\pi.
\end{align}
\end{subequations}

A natural choice for the free parameters
is the Lorentz gauge $c_1 = c_2 = c_3 = 1$. In this case
equations \eqref{eq:Maxwell_eq_of_motion}
are a first order in time second order in space formulation
of decoupled wave equations.
In what follows we will consider an arbitrary gauge choice $c_i$ 
and derive conditions that the resulting PDE system be symmetric 
hyperbolic.

In the notation of \eqref{eq:Ham_Matrix} we have
$u=(\p_iA^l,\p_i\phi,\pi_l,\pi)^\dag$ and
\begin{subequations}
\begin{align}
M^{-1 \,lk}      & = \left(\begin{array}{cc}
 \gamma^{lk}     & 0   \\
 0               & -c_2
\end{array}\right),\\
F^{i\,l}{}_{k}    & =  \left(\begin{array}{cc}
 0       &     -c_1\eta^{il} \\
 -c_3\delta^i{}_k & 0   
\end{array}\right),\\
V^{ij}{}_{lk} & =  \left(\begin{array}{cc}
 -\delta^{i}_{(l}\delta^{j}_{k)}+\eta^{ij}\eta_{lk} & 0   \\
 0       & 0   
\end{array}\right).
\end{align}
\label{eq:Maxwell_Candidate}
\end{subequations}
Note that the structure of the gauge sector prevents the evolution
equation of $\pi$ from containing terms like $\p_i\p^i\phi$, since the
evolution of the constraint subsystem must be closed. The index
structure of the variables also prevents the evolution equation for
$A^i$ from containing terms like $\p_i\p^i\phi$. These two facts
together guarantee the empty row and column of $V^{ij}$. In fact any
system with similar structure in the gauge have the same property. We 
will see this for the Einstein equations in the following sections.

The expanded Maxwell equations
\eqref{eq:Maxwell_eq_of_motion} can be written in a fully second order form 
provided that $c_2\ne 0$. One obtains
\begin{subequations}
\label{eq:fully_2nd_Maxwell}
\begin{align}
\nonumber
\p_t^2 A^i &= \p_j\p^jA^i+\left(c_3^2/c_2-1\right)\p^i\p_jA^j\\
&\qquad - \left(c_1c_2-c_3\right)/c_2 \p^i\p_t\phi,\\
\p_t^2\phi &= c_1^2c_2\p_i\p^i\phi + (c_1c_2-c_3)\p_i\p_t A^i,
\end{align}
\end{subequations}
and we read off the principal part
\begin{subequations}
\begin{align}
\mathcal{A}^{ij\,l}{}_{k}&=
\left(\begin{array}{cc}
\eta^{ij}\delta^{l}{}_k+(\frac{c_3^2}{c_2}-1)
 \eta^{l(i}\delta^{j)}{}_k & 0 \\ 
 0 & c_1^2c_2\eta^{ij}
\end{array}\right),\label{eqn:Maxwell_first}\\
  \mathcal{B}^{i\,l}{}_{k} &= \left(\begin{array}{cc}
    0 & -\frac{1}{c_2}(c_1c_2-c_3)\eta^{il} \\
    (c_1c_2-c_3)\delta^i{}_k & 0\end{array}\right).
\end{align}
\end{subequations}

\paragraph*{Strong hyperbolicity:}
Following \cite{Gundlach:2005ta} we investigate strong hyperbolicity
by performing a $2+1$ decomposition in space. The principal symbol then
decomposes into a scalar and a vector block. The characteristic speeds in
the vector block are $\pm 1$ and this submatrix is always diagonalizable.
In the scalar block of the expanded Maxwell equations the
characteristic speeds are $\pm \sqrt{c_1 c_3}$. Thus, strong 
hyperbolicity requires that $c_1$ and $c_3$ share a sign, or one of them
vanishes. Further calculations reveal that
the system is not strongly hyperbolic in the latter case.

Strong hyperbolicity additionally demands a complete set of 
characteristic variables. They exist only if $c_1c_2=c_3\ne 0$,
i.e. if in the fully second order system \eqref{eq:fully_2nd_Maxwell}
the vector and scalar potentials are decoupled. The characteristic
variables are given by
\begin{subequations}
\begin{align}
U_{s\pm v_1} &= \p_t\phi \pm c_1\sqrt{c_2}\p_s\phi,\\
U_{s\pm  v_2}&= \p_tA_s \pm c_1\sqrt{c_2}\p_sA_s,\\
U_{A\pm 1}&= \p_sA_A\pm \p_t A_A.
\end{align}
\end{subequations}
The non-trivial characteristic speeds are $v_1=v_2=\pm c_1\sqrt{c_2}$ so
strong hyperbolicity requires $c_2> 0$ as well.

\paragraph*{Symmetric hyperbolicity:}
Before we start the construction of positive symmetrizers we check whether
the necessary conditions for the existence of positive candidates, i.e. strong
hyperbolicity and the rank criterion, are satisfied.
Thus, we immediately demand strong hyperbolicity so $\mathcal{B}^i$
vanishes, $c_2>0$, $c_3=c_1c_2\ne 0$.
Application of the rank criterion (see section \ref{sec:rank_criterion})
results in a similar restriction. We find that the rank of the
relevant matrix $B(c)$ changes if and only if $c_3=c_1c_2$.

The remaining matrix~$\mathcal{A}^{ij}$
is automatically symmetric. So, taking $H_1=\mathbf{1}$, $\phi^i=0$,
$\phi^{ij}=0$ in lemma~\ref{lem:fully_2nd_sym_hyp},
the system is obviously symmetric hyperbolic if
$\mathcal{A}^{ij}$ has positive eigenvalues. Positivity of
the eigenvalues with $c_1c_2=c_3$ is equivalent to the condition
$\frac12 < c_1^2c_2 <3$. In geometric units the characteristic speeds $v$
of the system satisfy $1/\sqrt{2}<|v|<\sqrt{3}$.

Instead of the canonical candidate we may also consider the case where
\begin{align}
\label{eq:Maxwell_antisym_matrix}
\phi^{ij\,l}{}_k &=
\left(
\begin{array}{cc}
b_1 \eta^{l[i}\delta^{j]}{}_k & 0\\
0 & 0
\end{array}
\right) = \phi^{ij\,\dag}{}_k{}^l,
\end{align}
with some constant $b_1$. The eigenvalues of the resulting candidate
symmetrizer are
\begin{align}
\nonumber
\lambda_1 &= \frac{1 + c_1^2 c_2 - b_1}{2},&
\lambda_2 &= \frac{3 - c_1^2 c_2 + b_1}{2},\\
\nonumber
\lambda_3 &= -1 + 2c_1^2 c_2 + b_1,&
\lambda_4 &= c_1^2 c_2,\\
\lambda_5 &= 1.
\end{align}
Hence, if $c_2>0$ we can always choose $b_1$ such that this
candidate is positive definite. If $c_1^2 c_2\leq 4/3$ then
$b_1=1$ is a possible value and for $c_1^2 c_2 > 4/3$
one can take $b_1=c_1^2c_2$.
Thus, every strongly hyperbolic formulation of Electromagnetism with 
Hamiltonian structure is symmetric hyperbolic.


\subsection{The pure gauge system}
\label{section:pure_gauge}


In this section we consider scalar fields $(T,X^i)$ with second 
order equations of motion. We determine the set of equations 
of motion with Hamiltonian structure which are strongly 
and symmetric hyperbolic. Finally we consider which equations 
of motion obtained for the lapse and shift when the fields are 
taken as coordinates.

\paragraph*{Hamiltonian structure:}
We start by introducing canonical momenta $(\Theta,\Pi_i)$
for the scalar fields 
$(T,X^i)$ and write the most general ansatz for the Hamiltonian 
without introducing new geometric objects such that the equations
of motion are linear and naturally written with ``$\p_0=
\frac1\alpha(\p_t-\beta^i\p_i)$'' 
derivatives (where $\alpha$ and $\beta^i$ are lapse and shift
of the background manifold slicing respectively)
\begin{align}
\label{eq:PG_Hamiltonian}
\mathcal{H}_{D} &= \frac{c_1\alpha}{2\sqrt{\gamma}}\Theta^2
+\frac{c_2\alpha}{2}\sqrt{\gamma}D_iTD^iT
 +\frac{c_3\alpha}{2\sqrt{\gamma}}\Pi_i\Pi^i\\
& +\frac{c_4\alpha}{2}\sqrt{\gamma}D_iX^jD^iX_j
+\frac{c_5\alpha}{2}\sqrt{\gamma}D_iX^jD_jX^i\nonumber\\
\nonumber& + c_6 \alpha\Theta D_iX^i + c_7 \alpha\Pi^iD_iT
+ \Theta\beta^iD_i T + \Pi_i\beta^jD_j X^i.
\end{align}
Terms like $\Theta X^i D_iT$ would make the equations of motion 
nonlinear. We could also adjust the system by additional source 
terms. In the notation of \eqref{eq:Ham_Matrix} we have 
$u=(\p_jX^k,\p_iT,\Pi_l,\Theta)^\dag$ and
\begin{subequations}
\label{eq:PG_cand}
\begin{align}
\label{eq:matrix_Minv_pg}
M^{-1}{}^{lk} &= \frac{\alpha}{2\sqrt{\gamma}}\left(\begin{array}{cc}
c_3\gamma^{lk} & 0 \\
0 & c_1\end{array}\right),\\
\label{eq:matrix_F_pg}
F^{il}{}_{k} &= \frac{\alpha}{2}\left(\begin{array}{cc}
 0 & c_7\delta^{i}{}_k \\
c_6\gamma^{il} & 0\end{array}\right)
+ \frac{\beta^i}{2}\left(\begin{array}{cc}
 \delta^l_k & 0 \\
0 & 1\end{array}\right),\\
\label{eq:matrix_V_pg}
V^{ij}{}_{kl}&= \frac{\alpha\sqrt{\gamma}}{2}\left(\begin{array}{cc}
c_4\gamma^{ij}\gamma_{lk}+c_5\delta^{(i}{}_l\delta^{j)}{}_{k} & 0 \\
 0 &  c_2\gamma^{ij}\end{array}\right).
\end{align}
\end{subequations}
The pure gauge system does not carry any constraints, so unlike
in the Maxwell equations the $V^{ij}{}_{kl}$ matrix may contain
terms on both diagonal components. Otherwise the block structure 
guarantees that the the matrices look very similar to those 
of the Maxwell theory~\eqref{eq:Maxwell_Candidate}. In fact 
electromagnetism is a special case of the pure gauge system.

\paragraph*{Fully second order system:}
The number of parameters now present in the system complicates the 
analysis of hyperbolicity. To simplify the calculations we immediately 
switch to the fully second order version of the system, in which the 
equations of motion have principal matrices
\begin{subequations}
\begin{align}
\mathcal A^{ij} &= \alpha^2 \mathcal A_0^{ij} - 2 \alpha\beta^{(i}\mathcal B_0^{j)} - \beta^i\beta^j,\\
\mathcal B^i &= \alpha\mathcal B_0^i + 2\beta^i,
\end{align}
\end{subequations}
where the principal part in normal slicing ($\alpha=1$, $\beta^i=0$)
becomes
\begin{subequations}
\label{eq:FSO_pg}
\begin{align}
\label{eq:matrix_A_pg_2nd}
\mathcal{A}_0{}^{ijl}{}_{k}&=\left(\begin{array}{cc}
 \bar{c}_1\gamma^{ij}\delta^l{}_k + \bar{c}_2\gamma^{l(i}\delta^{j)}{}_k & 0 \\ 
 0   & \bar{c}_3\gamma^{ij}
\end{array}\right),\\
\label{eq:matrix_B_pg_2nd}
\mathcal{B}_0{}^{i}{}_{jl} &=\left(\begin{array}{cc}
 0 & \bar{c}_4\delta^i{}{}_j \\ 
 \bar{c}_5\delta^i{}_l & 0
\end{array}\right),
\end{align}
\end{subequations}
and we write
\begin{align}
\bar{c}_1 &= c_3c_4, &\quad \bar{c}_2=\frac{c_3}{c_1}(c_1c_5-c_6^2),\\
\bar{c}_3 &= \frac{c_1}{c_3}(c_2c_3-c_7^2),  &\quad 
\bar{c}_4=\frac{1}{c_1}(c_1c_7+c_3c_6),\\
\bar{c}_5 &= \frac{1}{c_3}(c_1c_7+c_3c_6), &
\end{align}
to simplify the matrices.
In what follows we assume normal slicing, i.e. $\alpha=1$,
$\beta^i=0$, $\mathcal A_0=\mathcal A$ and $\mathcal B_0=\mathcal B$.

Since we require $M^{-1}$ invertible, $c_1$ and 
$c_3$ must be nonzero, which means that if one of $\bar{c}_4$ and 
$\bar{c}_5$ vanish, they both must.

Given the linear fully second order system~\eqref{eq:FSO_pg} one can 
always derive a Hamiltonian formulation, provided 
$\bar c_4\neq 0\neq \bar c_5$, using
\begin{align*}
\nonumber
c_1 &= \frac{c_3\bar c_5}{\bar c_4},&
c_2 &= \frac{\bar c_3\bar c_4+\bar c_5 c_7^2}{c_3 \bar c_5},\\
c_4 &= \frac{\bar c_1}{c_3},&
c_6 &= \bar c_5\left(1-\frac{c_7}{\bar c_4}\right),
\end{align*}
and
\begin{align*}
c_5 &= \frac1{c_3\bar c_4}\left(\bar c_2\bar c_4+\bar c_4^2 \bar c_5 
-2 \bar c_4 \bar c_5 c_7 +\bar c_5 c_7^2\right),
\end{align*}
together with~\eqref{eq:PG_cand} where $c_3$ and $c_7$ can be chosen 
freely.

\paragraph*{Strong hyperbolicity:} We assume a generic choice of the parameters
$\bar c_1,\ldots,\bar c_5$, i.e. no degeneracies in the characteristic 
variables occur. Then we perform a $2+1$ decomposition in space and introduce 
the auxiliary quantities
\begin{align}
\lambda_X = \sqrt{\bar{c}_1+\bar{c}_2}, &\qquad \lambda_T = \sqrt{\bar{c}_3},
\end{align}
the speeds of the decoupled wave equations (in the scalar sector)
in normal slicing when 
$\mathcal{B}$ vanishes. In the vector sector the fully second order 
characteristic variables are
\begin{align}
U_{A\pm \lambda} & = \p_tX_A \pm \lambda_V \p_sX_A.
\end{align}
with speeds $\lambda_V = \sqrt{\bar{c}_1}$, so $\bar{c}_1>0$ is required.
With non-vanishing $\mathcal{B}$ the characteristic speeds in the scalar 
sector~$\pm(v_T,v_X)$ are modified to 
\begin{subequations}
\begin{align}
2v_T^2&= \lambda_T^2+\lambda_X^2+\bar{c}_4\bar{c}_5\nonumber\\
&+\sqrt{(\lambda_T^2-\lambda_X^2)^2+\bar{c}_4^2\bar{c}_5^2+
2\bar{c}_4\bar{c}_5(\lambda_X^2+\lambda_T^2)}    ,\\
2v_X^2&= \lambda_T^2+\lambda_X^2+\bar{c}_4\bar{c}_5\nonumber\\
&-\sqrt{(\lambda_T^2-\lambda_X^2)^2+\bar{c}_4^2\bar{c}_5^2+
2\bar{c}_4\bar{c}_5(\lambda_X^2+\lambda_T^2)}.
\end{align}
\end{subequations}
Note that if $\bar{c}_4$ or $\bar{c}_5$ vanish we recover the decoupled 
speeds. The characteristic variables in the scalar sector are 
\begin{subequations}
\begin{align}
U_{s\pm {v_X}}& = \p_tX_s\pm v_X\p_sX_s\nonumber\\
&\pm\frac{\bar{c}_5v_X}{v_X^2-\lambda_T^2}\big(\p_tT\pm v_X\p_sT\big),\\
U_{s\pm {v_T}}& = \p_tT\pm v_T\p_sT\nonumber\\
&\pm\frac{\bar{c}_4v_T}{v_T^2-\lambda_X^2}
\big(\p_tX_s\pm v_T\p_sX_s\big),
\end{align}
\end{subequations}
when $v_X\ne\lambda_T$ and $v_T\ne\lambda_X$. Further calculation 
using the fact that $v_X^2v_T^2=\lambda_X^2\lambda_T^2$ reveals that 
$v_X = \lambda_T$ if and only if $v_T = \lambda_X$, which is only 
possible if $\bar{c}_4$ and $\bar{c}_5$ vanish. In this case the 
variables decouple and the only condition for strong hyperbolicity 
is that the speeds are real. In summary, the system is strongly 
hyperbolic provided that
\begin{align}
\lambda_V^2=\bar{c}_1 & > 0, 
\quad  & \lambda_X^2 = \bar{c}_1+\bar{c}_2 \ne 0,
\nonumber\\
\lambda_T^2=\bar{c}_3 & \ne 0, 
\quad  & v_T^2 > 0,\nonumber\\
 v_X^2    &  >  0. \quad  & 
\end{align}
To obtain the characteristic variables for the fully
second order system in a general slicing from the expressions in
normal slicing one can e.g. simply replace the $\p_t$ derivative
in those expressions by a ``$\p_0$'' derivative:
\begin{align}
\p_t X \rightarrow \frac1\alpha\left(\p_t X - \beta^s \p_s X\right).
\end{align}
E.g. in the vector sector one obtains
\begin{align}
U_{A\pm \lambda} & = \frac1\alpha\left(\p_tX_A - (\beta^s \mp \alpha\lambda) \p_sX_A\right).
\end{align}
Given a characteristic speed $v$ in normal slicing it becomes
$\alpha v + \beta^s$ for general slicings.

\paragraph*{Symmetric hyperbolicity:}
As in the case of electromagnetism we first demand strong hyperbolicity.
The second necessary condition, the rank criterion, cannot be applied
here, because in the parameter space we do not find a dense
set of formulations which are not symmetric hyperbolic.

We ask whether the modified canonical candidate
\begin{align}
\bar H^{ij}{}_{kl}=\left(
\begin{array}{cc}
M_{km}\mathcal A^{ij\,m}{}_l + \phi^{[ij]}{}_{[kl]} & 0\\
0 & M_{kl}
\end{array}
\right),
\end{align}
with the matrix $\phi^{ij}{}_{kl}$ defined in \eqref{eq:Maxwell_antisym_matrix}
is positive. We immediately get that
\begin{align}
c_3 &> 0,&
\bar c_4/\bar c_5 &> 0,&
\bar c_3>0
\end{align}
must be satisfied, because the candidate is block diagonal with one 
diagonal block which has those elements.

The eigenvalues of the remaining upper left block are
\begin{align}
b_1 - \frac{\bar c_2 - 2 \bar c_1}{2 c_3},&&
\frac{2 \bar c_1 + \bar c_2}{2 c_3} - b_1,&&
\frac{\bar c_1 + 2 \bar c_2}{c_3} + 2 b_1.
\end{align}
Hence, we can choose $b_1$ such that the modified candidate becomes positive
provided that
\begin{align}
\bar c_1 &> 0,&
\bar c_1+\bar c_2 &> 0.
\end{align}
Strong hyperbolicity of the system only implies that
$\bar c_1>0$, $\lambda_T^2\lambda_X^2>0$ and
$(\lambda_T^2-\lambda_X^2)^2+\bar{c}_4^2\bar{c}_5^2+
2\bar{c}_4\bar{c}_5(\lambda_X^2+\lambda_T^2)>0$.
Hence, as opposed to Electromagnetism, we get strongly 
hyperbolic formulations where the modified canonical candidate 
is not positive definite.

However, here the expressions are still simple enough to consider
all matrices that can be constructed from the metric. We derive 
general candidate symmetrizers in that class and analyze their
positivity. We find that for a generic choice of the parameters
every strongly hyperbolic formulation is also symmetric hyperbolic.

\paragraph*{Discussion:} It is not clear by examining 
matrices~(\ref{eq:FSO_pg}) exactly what restriction Hamiltonian 
structure puts upon the pure gauge system. The only condition 
is that when either of $\bar{c}_4$ or $\bar{c}_5$ vanish then 
they both do. One may analyze the level of hyperbolicity of the 
system even if Hamiltonian structure is abandoned. The two cases 
$\bar{c}_4=0$ and $\bar{c}_5=0$ are similar, so we describe only 
the first. We find that whenever $\lambda_X\ne\lambda_T$ the 
system is both strongly and symmetric hyperbolic whenever $\lambda_T,
\lambda_X>0$.

With this example we demonstrate that there are strongly hyperbolic
Hamiltonian formulations where apparently the existence of the
canonical candidate does not simplify the symmetric hyperbolicity
analysis. Moreover, in the fully second order form of the equations
of motion Hamiltonian structure does not play a role. Thus, if general
relativity is considered then we do not expect to find symmetrizers
through a simple modification of the canonical candidate.

\paragraph*{Evolution of the lapse and shift:}
To derive equations of motion for the lapse and shift from 
our pure gauge system we may simply use the relationship between the 
lapse and shift $\alpha,\beta^i$ and the local coordinates
in time and space $(t,x^i$)
\begin{align}
\alpha&= -\frac{1}{n^a\p_at}, \qquad \beta^i= -\alpha n^a\p_ax^i,
\end{align}
with unit normal to the hypersurface $n_a$, and take $T=t$, $X^i=x^i$. 
The lapse and shift evolution equations do not inherit the same 
principal part as the pure gauge system from which they were 
derived. Furthermore, it is not possible to give the fully second 
order equations of motion for the lapse and shift until the gauge 
is coupled to particular equations of motion for the metric and 
extrinsic curvature.
 

\section{Formulations in numerical GR}
\label{section:ADM_Gauge}


In this section we investigate formulations of GR with 
Hamiltonian structure. The aim is to derive a symmetric 
hyperbolic Hamiltonian formulation with properties that 
are believed to be important for stable numerical simulations 
and to relate this formulation to currently existing ones.


\subsection{The ADM system}
\label{section:ADM}


The starting point in the derivation of a new formulation is 
the ADM system \cite{Arnowitt:1962hi}. It can be derived 
through a $3+1$ decomposition of the vacuum Einstein equations 
\cite{York}. If one writes the result in geometric variables  
$(\gamma_{ij},K_{ij})$, where $\gamma_{ij}$ is the 3-metric and 
$K_{ij}$ is the extrinsic curvature induced in a spacelike slice, 
then one obtains equations of motion
\begin{subequations}
\label{eq:dynamical_ADM}
\begin{align}
\p_t\gamma_{ij}&=-2\alpha K_{ij}+\Lie_\beta\gamma_{ij},
\label{eq:definition_K}\\
\p_t K_{ij}&= -D_iD_j\alpha+\alpha[R_{ij}-2K_{ik}K^k_j+K_{ij}K]
\nonumber\\
\label{eq:dynamical_ADM_b}
&\qquad+\Lie_\beta K_{ij}.
\end{align}
\end{subequations}
Equation~\eqref{eq:definition_K}
defines $K_{ij}$ and eqn.~\eqref{eq:dynamical_ADM_b} is the projection 
of the Einstein equations out of the slice. By contracting into the slice 
one also obtains the Hamiltonian $C$ and momentum $C_i$ constraints, which 
are given by 
\begin{subequations}
\label{eq:ADM_constraints}
\begin{align}
C&=R+K^2-K_{ij}K^{ij}=0,\\
C_i&=D_j(K^{ij}-\gamma^{ij} K)=0.
\end{align}
\end{subequations}
The ADM equations can also be written in terms of 
\emph{canonical ADM variables}, $(\gamma_{ij};\pi^{ij})$,
through the relation
\begin{align}
\label{eq:rel_pi_K}
\pi^{ij}&=\sqrt{\gamma}(K\gamma^{ij}-K^{ij}).
\end{align}
In these variables the dynamical ADM equations 
\eqref{eq:dynamical_ADM} have a Hamiltonian structure with 
Hamiltonian
\begin{align}
{\cal H}_{\textrm{ADM}}=
\int(-{\alpha} C+2\beta^iC_i)\sqrt{\gamma}
{\textrm d}^3x,\label{eqn:H_ADM}
\end{align}
where the constraints must be rewritten in terms of the
canonical ADM variables.

The constraint evolution system is closed, so in the absence 
of a spatial boundary they need only be satisfied in one time-slice 
to guarantee their satisfaction later in time. This justifies 
analyzing the PDE properties of the system without assuming 
the constraints. When a spatial boundary is present suitable constraint 
preserving boundary conditions are also required.

In the ADM formulation the lapse $\alpha$ and shift $\beta^i$
are thought of as given fields. The formulation is weakly, but 
not strongly hyperbolic (the principal symbol has real eigenvalues). 
The initial value problem is not well-posed.


\subsection{Gauge conditions}
\label{section:Pure_Gauge}


In the introduction we noted that from the free-evolution PDEs
point of view, the gauge freedom of GR has two aspects, the 
first of which is the ability to choose arbitrary equations 
of motion for the lapse and shift. Choices in which either 
the lapse or shift evolve dynamically are known as live gauge
conditions. In this section we describe two popular live gauge 
conditions. 

\paragraph*{Harmonic gauge:} In terms of the lapse and shift the 
harmonic gauge condition $\Box X^a=0$ becomes
\begin{subequations}
\begin{align}
\p_t\alpha&=-\alpha^2K+\beta^p\alpha_{,p},\\
\p_t\beta_i&=\beta^j\p_j\beta^i-\alpha \p^i\alpha
+\alpha^2\gamma^{jk}{\Gamma^i}_{jk}.
\end{align}
\end{subequations}
The use of this gauge was the key component in the first proofs 
that GR admits a well-posed initial value 
problem~\cite{Foures-Bruhat:1952}, since when it is appropriately 
coupled to the EE's they become a particularly simple symmetric 
hyperbolic system of wave equations. When the gauge is modified 
with arbitrary source terms hyperbolicity of that formulation is 
not affected~\cite{Garfinkle:2001ni}. The modified gauge is known 
as generalized harmonic gauge or GHG.

\paragraph*{Puncture gauge:} We call the combination of the 
popular Bona-Mass\'o lapse and gamma driver shift conditions
\begin{subequations}
\label{eq:puncture_gauge}
\begin{align}
\p_t\alpha&=-\mu_L\alpha^2K+\beta^p\alpha_{,p},
\label{eqn:BM_gauge}\\
\p_t\beta_i&=\mu_S\tilde{\Gamma}^i+\beta^j\p_j\beta^i-\eta\beta.
\label{eqn:Gamma_driver}
\end{align}
\end{subequations}
with
\begin{align}
\tilde{\Gamma}^i&=\gamma^{1/3}\gamma^{jk}\Gamma^{i}_{jk}
+\frac{1}{3}\gamma^{1/3}\gamma^{ji}\Gamma^k{}_{kj}
\end{align}
the puncture gauge. In fact there are various choices of 
parameters inside the conditions which we do not consider here; 
for a review of the choices made by various groups see table 1 
of~\cite{Gundlach:2006tw}. Our gamma driver 
condition~\eqref{eqn:Gamma_driver} can be obtained by integrating 
the most commonly implemented condition once in 
time~\cite{vanMeter:2006vi}. Stationary data for the lapse 
condition was discussed in~\cite{Hannam:2006vv,Brown:2007nt,
Garfinkle:2007yt}.


\subsection{Hamiltonian formulations}
\label{sec:Hamiltonian_formulations}


Now we are interested in whether or not strongly or
symmetric hyperbolic Hamiltonian formulations with a large 
class of live gauge conditions can be derived. Our 
Hamiltonian formulations are described in terms of the 
\emph{canonical variables} 
$(\gamma_{ij},\alpha,\beta^i;\pi^{ij},\sigma,\rho_i)$.
For solutions of the Einstein equations $\sigma$ and $\rho_i$ are
constrained to vanish.

To get a Hamiltonian formulation in other variables, e.g. densitized
ones, one can apply a symplectic transformation to obtain the correct
Hamiltonian. This does not alter hyperbolicity. In general, if the
desired position variables are $f(q)$ then the corresponding canonical
momenta are $f'(q)^{-\dag}p$, where $f'$ is the Jacobi matrix of $f$
($f$ must not contain derivatives of $q$).

To derive formulations we follow the approach of 
Brown~\cite{Brown:2008cca}, i.e. we add appropriate gauge terms 
to the Super-Hamiltonian $\mathcal H_{\mathrm{ADM}}$:
\begin{align}
\label{eq:general_Hamiltonian}
\mathcal H = \mathcal H_{\mathrm{ADM}} + \mathcal H_{\mathrm GHG}
+ \int d^3x \left(\Lambda\sigma +\Omega^i \rho_i\right),
\end{align}
where $\mathcal H_{\mathrm{ADM}} + \mathcal H_{\mathrm GHG}$ is the
Hamiltonian for Brown's generalized harmonic formulation \cite{Brown:2008cca}.
The terms that appear in the principal part are
\begin{align}
\nonumber
&\mathcal H_{\mathrm GHG} =
\beta^{i}\sigma D_{i}\alpha
+ \beta^{i}\rho_j\p_i\beta^j
+ \alpha^2\rho_i\Gamma^i_{jk}\gamma^{jk} - \alpha\rho^{i}D_{i}\alpha \\
&\qquad
+ \frac1{8\sqrt{\gamma}}
\big(-4 \alpha^2 \gamma_{ij}\pi^{ij}\sigma
 + \alpha^3 \sigma^2
 - 4 \alpha^3 \rho_{i} \rho_{j}\gamma^{ij}\big).
\end{align}
The terms in $\Lambda$ and $\Omega^i$ that affect the principal part of
the equations of motion are linear in the canonical momenta
and in the first spatial derivatives of the positions.

We restrict the possible terms further by considering those gauges 
for which the shift appears in the principal part only through its 
spatial derivatives and the obvious advection terms. 

This is of course a restriction of the admissible gauges. But the live gauge
conditions that are used in numerical GR are usually of this form. Moreover in
those gauges the hyperbolicity analysis is simpler, because, as we will
see, one may linearize around Minkowski space without loss of generality.
Furthermore the discussion of Sect.~\ref{sec:Decoupled_variables} is 
applicable.

One also finds strongly hyperbolic formulations without this ``$\p_0$" 
restriction, but they are not considered here. We get the following 
expressions
{
\allowdisplaybreaks
\begin{align}
\nonumber
\Lambda &=
-C_1 \alpha^2 \gamma^{-1/2} \gamma_{ij} \pi^{ij} +
C_4 \alpha^3 \gamma^{-1/2} \sigma+\\
&\qquad +
C_7\left(\alpha D_i\beta^i - \frac12\alpha \gamma^{jk} \beta^i \p_i\gamma_{jk}\right)
\end{align}
and
\begin{align}
\Omega^i &=
C_2 \alpha^2 \Gamma^i_{jk} \gamma^{jk} +
C_3 \alpha^2 \Gamma^k_{kj} \gamma^{ji}+\\
\nonumber&\qquad -
C_5 \alpha \gamma^{ij} D_j\alpha
- C_6 \gamma^{-1/2} \alpha^3 \rho_j \gamma^{ij}.
\end{align}
}
We denote $C_0,\ldots,C_{7}$ \emph{formulation parameters}
and obtain Brown's Hamiltonian \cite{Brown:2008cca} when $C_7=0$
and we replace
\begin{align}
\nonumber
C_1&\rightarrow C_1-1/2,&
C_2&\rightarrow C_2-1,&
C_4&\rightarrow C_4-1/8,\\
C_5&\rightarrow C_5-1,&
C_6&\rightarrow C_6-1/2.
\end{align}

For those formulations we get in the notation of \eqref{eq:Ham_Matrix} that
$u=(\p_i\gamma_{km},\p_i\alpha,\p_i\beta^k,\pi^{km},\sigma,\rho_k)^\dag$
and
\begin{subequations}
\begin{widetext}
\begin{align}
\label{eq:C1free_PP}
M^{-1}_{kl\,mn} &=\alpha\gamma^{-1/2}\left(
\begin{array}{ccc}
2\gamma_{k(m}\gamma_{n)l} - \gamma_{kl}\gamma_{mn} &
- (1/2+C_1)\alpha\gamma_{kl} &
0\\
-(1/2+C_1)\alpha\gamma_{mn} & (1/4+2C_4)\alpha^2 &
0\\
0&
0&
-(1+2C_6)\alpha^2\gamma^{km}
\end{array}
\right),\\
F^{j}{}_{kl}{}^{mn} &=\left(
\begin{array}{ccc}
\beta^j\delta^{(m}_l\delta^{n)}_k & 0 & 2\delta^j_{(l}\gamma_{k)m}\\
0& \beta^j &
C_7\alpha\delta^j_m\\
\alpha^2\left((1+C_2)\gamma^{j(m}\gamma^{n)k}-
\frac12(C_2-C_3+1)\gamma^{jk}\gamma^{mn}\right) & 
-(1+C_5)\alpha\gamma^{jk} &
 \beta^j\delta_m^k
\end{array}
\right),\\
V^{ij\,kl\,mn} &=\gamma^{1/2}\left(
\begin{array}{ccc}
V_{11}^{ijklmn} & \gamma^{i(k}\gamma^{l)j} - \gamma^{ij}\gamma^{kl} & 0\\
\gamma^{i(m}\gamma^{n)j} - \gamma^{ij}\gamma^{mn} & 0 & 0\\
0 & 0 & 0
\end{array}
\right),
\end{align}
\end{widetext}
\end{subequations}
where
\begin{align*}
V_{11}^{ij\,kl\,mn} &= -\alpha\big(
\gamma^{ij}\gamma^{kl}\gamma^{mn}
+\gamma^{i(k}\gamma^{l)(m}\gamma^{n)j}\\
&\qquad
+\gamma^{i(m}\gamma^{n)(k}\gamma^{l)j}
-\gamma^{ij}\gamma^{k(m}\gamma^{n)l}\\
&\qquad
-\gamma^{kl}\gamma^{i(m}\gamma^{n)j}
-\gamma^{mn}\gamma^{i(k}\gamma^{l)j}
\big)/2.
\end{align*}
Various entries of $V^{ij\,kl\,mn}$ vanish for exactly the same 
reason as those in the Maxwell system~\eqref{eq:Maxwell_Candidate}.

To see that for the hyperbolicity analysis one may linearize around
Minkowski space observe that 
when $A^p(\det(\gamma_{ij}),\alpha,\beta^i)$ is the first order in time, 
second order in space principal part for an arbitrary background 
solution and
\begin{align}
\label{eq:sim_matrix}
\nonumber
T=\mbox{diag}(&\alpha^{-1/2}\gamma^{-1/4},\alpha^{1/2}\gamma^{-1/4},
\alpha^{1/2}\gamma^{-1/4},\\
&\alpha^{-1/2}\gamma^{1/4},\alpha^{-3/2}\gamma^{1/4},
\alpha^{-3/2}\gamma^{1/4})
\end{align}
then
\begin{align}
T^{-1}A^p(\det(\gamma_{ij}),\alpha,\beta^i)T = 
\alpha A^p(1,1,0)+\beta^p \mathbf 1.
\end{align}
Diagonalizability and eigenvalues are not altered by a similarity
transformation, i.e. to analyze strong hyperbolicity one may consider
the matrix $s_pA^p(1,1,0)$ for an arbitrary spatial vector $s$.

Concerning symmetric hyperbolicity a short calculation shows that when $H$ is
a symmetrizer of $A^p(1,1,0)$ then $T^{\dag}H T$ is a symmetrizer of
$A^p(\det(\gamma_{ij}),\alpha,\beta^i)$ for an arbitrary background metric.
In what follows we assume $\alpha=1=\gamma$ and $\beta^i=0$.

\paragraph*{Strong hyperbolicity:}
Again we perform a $2+1$ decomposition in space \cite{Gundlach:2005ta}, and the
principal symbol decomposes into three submatrices, a scalar, vector and
trace-free tensor block. The characteristic speeds in the trace-free tensor
and vector sector are
\begin{align}
\label{eq:tensor_vector_speeds}
\lambda_{TF}^2 &= 1,&
\lambda_V^2 &= 1+C_2
\end{align}
respectively. In the scalar block one finds
\begin{subequations}
\label{eq:Speeds_v+-}
\begin{align}
\label{eq:Speed_v+}
2v_+^2&=\lambda_+^2+\lambda_-^2-C_5C_7\\
\nonumber
&+\sqrt{(\lambda_+^2-\lambda_-^2)^2+C_5^2C_7^2
-2C_5C_7(\lambda_+^2+\lambda_-^2)},\\
\label{eq:Speed_v-}
2v_-^2&=\lambda_+^2+\lambda_-^2-C_5C_7\\
\nonumber
&-\sqrt{(\lambda_+^2-\lambda_-^2)^2+C_5^2C_7^2
-2C_5C_7(\lambda_+^2+\lambda_-^2)},
\end{align}
\end{subequations}
where
\begin{align}
\label{eq:defn_lam_pm}
2\lambda_\pm^2&=(2C_1+C_2+C_3-C_7+2)\nonumber\\
&\qquad\pm(-2C_1+C_2+C_3+C_7).
\end{align}
Note that the speeds in the scalar sector are very similar to those 
of the pure gauge system, and simplify significantly in the special 
case $C_5=0$ or $C_7=0$ (or both). They simplify further when 
$\lambda_+=\lambda_-$.

Using the techniques described in 
App.~\ref{app:strong_hyp_cond} we can classify the formulations with 
respect to their strong hyperbolicity. We find three families 
${\mathcal F}_{i=1,2,3}$, of strongly hyperbolic fully second order 
formulations.

From the vector block of the principal symbol it is obvious that the 
condition
\begin{align}
\label{eq:vector_conditions}
C_6 &= (\lambda_V^2-1)/2
\end{align}
must be satisfied in every case. Moreover it is clear that the 
characteristic speeds must be real, i.e.
\begin{align}
\lambda_V^2&\geq 0, &
v_{+/-}^2&\geq 0.
\end{align}
The three families correspond to different solutions of the requirement
that the scalar block is diagonalizable.

\paragraph*{Family $\mathcal{F}_1$:} The first family has four free parameters, 
$(\lambda_+,\lambda_-,C_5,C_7)$, and one obtains strongly hyperbolic 
formulations if $\lambda_-\ne 0$, $v_+^2\neq v_-^2$,
\begin{subequations}
\label{eq:scalar_four_param}
\begin{align}
\label{eq:scalar_four_param_a}
0 &\neq (1+2C_5)\lambda_-^2 - \lambda_+^2,\\
\label{eq:scalar_four_param_b}
2\lambda_V^2&=3\lambda_+^2-1
+\frac{C_5\lambda_-^2+3C_7\lambda_+^2}
{(1+2C_5)\lambda_-^2 - \lambda_+^2}C_5,\\
\label{eq:scalar_four_param_c}
8C_4&=\lambda_-^2-1+\frac{C_5\lambda_-^2+3C_7\lambda_+^2}
{(1+2C_5)\lambda_-^2-\lambda_+^2}C_7.
\end{align}
\end{subequations}
where we have replaced $C_1$, $C_2$ and $C_3$ using the 
definitions~\eqref{eq:defn_lam_pm} and \eqref{eq:tensor_vector_speeds} 
of $\lambda_{\pm}$ and $\lambda_V$.

\paragraph*{Family $\mathcal{F}_2$:} 
If the inequality \eqref{eq:scalar_four_param_a} is not satisfied,
but $\lambda_-\neq 0$, $v_+^2\neq v_-^2$ still hold, then the 
solution~\eqref{eq:scalar_four_param} becomes singular. One obtains 
another family with three free parameters, $(\lambda_-,\lambda_V,C_7)$:
\begin{subequations}
\label{eq:scalar_three_param}
\begin{align}
12C_4 &=
3\lambda_-^2 - 3 - 9 C_7 - (1+6C_7)(\lambda_V^2-1),\\
C_5 &= -\frac{3 C_7}{1 + 6 C_7},
\quad 
\lambda_+^2 = \frac{\lambda_-^2}{1+6C_7}.
\end{align}
It must satisfy
\begin{align}
C_7 &\neq 0,&
C_7 &\neq -1/6.
\end{align}
\end{subequations}
The first of these inequalities guarantees that there are no more 
repeated speeds and the second prevents the solution from becoming
singular, both of which are special cases that contain no further 
strongly hyperbolic formulations. The characteristic speeds in the 
scalar sector become
\begin{align}
\nonumber
2(1+6C_7)v_{\pm}^2 &=  3C_7(2\lambda_-^2+C_7)+2\lambda_-^2\\
&\qquad \pm C_7\sqrt{9(2\lambda_-^2+C_7)^2+12\lambda_-^2}.
\end{align}

\paragraph*{Family $\mathcal{F}_3$:} If the scalar block has only one 
eigenvalue ($v_+^2 = v_-^2$) one obtains a one parameter family of 
formulations with $\lambda_+>0$:
\begin{align}
\label{eq:scalar_one_param_equal_speeds}
\nonumber
\lambda_-&=\lambda_+,&C_5 &= 0,& C_7 &= 0,\\
2\lambda_V^2 &= 3\lambda_+^2-1,& C_4 &= 1/8(\lambda_+^2-1).&&
\end{align}
This is the limit of $\mathcal F_1$ \eqref{eq:scalar_four_param}
for $\lambda_+\to\lambda_-$ along curves with $C_5= 0 = C_7$.
But if one takes in \eqref{eq:scalar_four_param} the limit 
$v_+\rightarrow v_-$ then the result is in general not strongly 
hyperbolic. To obtain \eqref{eq:scalar_one_param_equal_speeds} one 
may also take the limit $C_7\rightarrow0$, 
$\lambda_V^2\rightarrow (3\lambda_-^2-1)/2$ of $\mathcal F_2$
\eqref{eq:scalar_three_param} 
(in any order). Again the case $C_7=0$ is not strongly hyperbolic for 
general choices of $\lambda_V$.

\paragraph*{Fully second order system:}
The fully second order equations of motion can be constructed for 
each of three strongly hyperbolic systems $\mathcal{F}_i$. We present
only $\mathcal{F}_3$. In terms of variables 
$u_{kl}=(\gamma_{kl},\alpha,\beta_k)^\dag$ the two parameter strongly 
hyperbolic system $\mathcal{F}_3$ has the fully second order principal 
part
\begin{align}
\p_t^2u_{kl}&= \mathcal A^{ij}{}_{kl}{}^{mn}\p_i\p_ju_{mn}
+\mathcal B^{i}{}_{kl}{}^{mn}\p_t\p_iu_{mn}.
\end{align}
The principal matrices $\mathcal A^{ij}{}_{kl}{}^{mn}$ and 
$\mathcal B^{i}{}_{kl}{}^{mn}$ are 
\begin{subequations}
\begin{widetext}
\begin{align}
\label{eq:AijBi_Einstein}
\mathcal A^{ij}{}_{kl}{}^{mn} &=
\left(
\begin{array}{ccc}
\mathcal A_{11}^{ij}{}_{kl}{}^{mn} &
\bar C_5\gamma^{ij}\gamma_{kl} & 0\\
0 & \bar C_{6}\gamma^{ij} & 0\\
0 & 0 & \bar C_{7} \gamma^{ij}\delta_k^m 
+ \bar C_{8} \gamma^{m(i}\delta^{j)}_k
\end{array}
\right),\quad
\mathcal B^{i}{}_{kl}{}^{mn} &=
\left(
\begin{array}{ccc}
0 & 0 &
\bar B_1\delta^i_m\gamma_{kl}\\
0 & 0 & 0                   \\
\bar B_2 \gamma^{ik}\gamma^{mn} &
\bar B_3 \gamma^{ik} &
0
\end{array}
\right),
\end{align}
\end{widetext}
\end{subequations}
where
\begin{align}
\nonumber
\mathcal A_{11}^{ij}{}_{kl}{}^{mn} &=
\gamma^{ij}\delta_{(k}^m\delta^{n}_{l)}
+\bar C_1\gamma^{i(m}\gamma^{n)j}\gamma_{kl}\\
\nonumber &\qquad
+\bar C_2\gamma^{ij}\gamma_{kl}\gamma^{mn}
+\bar C_3\delta^{i}_{(k}\delta_{l)}^{j}\gamma^{mn}\\
&\qquad
+\bar C_4\left(\delta^{i}_{(k}\delta_{l)}^{(m}\gamma^{n)j}
+ \delta^{j}_{(k}\delta_{l)}^{(m}\gamma^{n)i}\right).
\end{align}
The parameters in $\mathcal A_{11}^{ij}{}_{kl}{}^{mn}$  
are given by
\begin{subequations}
\begin{align}
2\bar C_1 &= 1-\lambda_+^2,&\quad
2\bar C_2 &= \lambda_+^2-1
+2\frac{\lambda_-^2-\lambda_+^2}{1+3\lambda_+^2}\lambda_+^2,\\
2\bar C_3 &=1-\lambda_+^2,&\quad
4\bar C_4 &=3(\lambda_+^2-1).
\end{align}
\end{subequations}
The remaining parameters in $A^{ij}{}_{kl}{}^{mn}$ are
\begin{subequations}
\begin{align}
\bar C_5  &= 2\frac{\lambda_+^2-\lambda_-^2}{1+3\lambda_+^2},&\quad
\bar C_6  &= \lambda_-^2,\\
4\bar C_7  &= 3\lambda_+^2+1,&\quad
4\bar C_8  &= \frac{\lambda_-^2+1}{3\lambda_-^2+1}(3\lambda_+^2+1).
\end{align}
\end{subequations}
The parameters of $\mathcal B^{i}{}_{kl}{}^{mn}$ are 
\begin{subequations}
\begin{align}
\bar B_1 &= 2\frac{\lambda_+^2-\lambda_-^2}
{1+3\lambda_+^2},&\quad
2\bar B_2 &= \frac{\lambda_+^2-\lambda_-^2}
{1+3\lambda_+^2},\\
\bar B_3 &=3\frac{\lambda_+^2-\lambda_-^2}
{1+3\lambda_+^2}.&&
\end{align}
\end{subequations}

\paragraph*{Characteristic variables:}
The characteristic variables may be constructed for each 
family $\mathcal{F}_i$. For brevity we present them only for the 
subset $C_5=C_7=0$ of the first family $\mathcal{F}_1$. The 
characteristic variables in the scalar sector are 
\begin{subequations}
\begin{align}
U_{S\alpha\pm \lambda_-} &= (\p_t\pm\lambda_-\p_s)\alpha\mp
2\lambda_-(\p_t\pm\lambda_-\p_s)\gamma_{ss}\nonumber\\
&\mp 4\lambda_-(\p_t\pm\lambda_-\p_s)\gamma_{qq},\\
U_{S\beta\pm\lambda_-}&= (\p_t\pm\lambda_-\p_s)\alpha\mp
2\lambda_-(\p_t\pm\lambda_-\p_s)\gamma_{ss}\nonumber\\
&\mp 4\lambda_-(\p_t\pm\lambda_-\p_s)\gamma_{qq},\\
2U_{S\beta\pm\lambda_+} &= 2(\p_t\pm\lambda_+\p_s)\gamma_{ss}
\mp\lambda_+(\p_t\pm\lambda_+\p_s)\beta_{s},\\
U_{S\pm\lambda_+} &= (\p_t\pm\lambda_+\p_s)
(\gamma_{ss}-\gamma_{qq}),
\end{align}
\end{subequations}
with speeds $\lambda_+$ and $\lambda_-$ as given 
above~\eqref{eq:Speeds_v+-}. The vector sector has 
\begin{subequations}
\begin{align}
U_{A\pm\gamma}&= (\p_t\pm\lambda_V\p_s)\gamma_{As},\\
U_{A\pm\beta}&= (\p_t\pm\lambda_V\p_s)\beta_{A}.
\end{align}
\end{subequations}
with $\lambda_V=\sqrt{1+C_2}$. Finally the tensor sector has characteristic 
variables 
\begin{align}
U_{AB\pm 1}^{TF}&= (\p_t\pm\p_s)\gamma_{AB}^{TF},
\label{eq:GR_Tensor_CharVar}
\end{align}
where $\gamma_{AB}^{TF}$ denotes the transverse trace free part of the 
metric against $s_i$. 

To find the characteristic variables of the system in an arbitrary
background, one must transform the eigenvectors of the principal 
symbol through the similarity transformation defined by $T$ 
from \eqref{eq:sim_matrix} and replace the time 
derivatives $\p_t\to\p_0$.

\paragraph*{Symmetric hyperbolicity:}
As for the Maxwell equations we impose strong hyperbolicity and check
where the rank criterion is satisfied. Interestingly we find again that
the rank of the relevant matrix $B(c)$ from section \ref{sec:rank_criterion}
is discontinuous (and drops) at every strongly hyperbolic formulation. Hence,
this does not lead to immediate restrictions in the parameter space.

A complete analysis of symmetric hyperbolicity seems to be too hard
in the general case, because expressions become very complicated.
We want to discuss special cases where this problem 
can be avoided and conjecture about the general case.

In App.~\ref{app:symmetrizer_constr} we construct positive definite 
symmetrizers for formulations whose parameters satisfy
\begin{align}
\label{eq:sym_hyp_form}
\nonumber
\lambda_V^2 &= (3\lambda_+^2+1)/4,&
C_4 &= (\lambda_-^2-1)/8,&
C_5 &= 0,\\
\nonumber
C_6 &= 3(\lambda_+^2-1)/8,&
C_7 &= 0,\\
\lambda_-^2 &> 0,&
\lambda_+^2 &> 0.
\end{align}
Hence, we found a two parameter family of symmetric hyperbolic
formulations. Using the same techniques we are also able to prove
symmetric hyperbolicity of other two parameter families. With
\eqref{eq:sym_hyp_form} they share that they belong to the
four parameter family $\mathcal F_1$
\eqref{eq:scalar_four_param} and that $C_5=0$.

For the three parameter subfamily of $\mathcal F_1$
that satisfies also $C_5=0$ we can construct candidates $H^{ij}$ such that the contractions
$s_i H^{ij} s_j$ are positive definite for every spatial vector $s$.
We think that the full matrices $H^{ij}$ can be chosen positive
definite in this case, but are not able to prove it.

We also deal with the question whether also in GR every strongly
hyperbolic formulation  is symmetric hyperbolic, as it was the case
for the toy problems in section~\ref{section:Maxwell_Pure_Gauge}.
We find that the (strongly hyperbolic)
two parameter subfamily family of $\mathcal F_1$ with
\begin{align}
\label{eq:non_sym_hyp_form}
\nonumber
 \lambda_V^2 &= \frac14(1 + C_5)(1 + 3\lambda_-^2),&
 \lambda_+^2 &= (1+C_5)\lambda_-^2,\\
\nonumber
 C_4 &= \frac18(\lambda_-^2-1),\\
\nonumber
 C_6 &= \frac18(C_5 (3\lambda_-^2+1) + 3(\lambda_-^2-1)),&
 C_7 &= 0,\\
\lambda_-^2&>0,& C_5&>-1
\end{align}
and a generic choice of $\lambda_-^2$ and $C_5$
is not symmetric hyperbolic.

To prove that we first search for candidates in the set of matrices
that can be written in terms of the metric. We find a six dimensional
space, $\mathcal G=\mathcal G_S\oplus\mathcal G_A$, which decomposes into
a four dimensional $\mathcal G_S$, the symmetric part $\bar H^{(ij)}$,
and a two dimensional $\mathcal G_A$, the antisymmetric $\phi^{[ij]}$
\eqref{eq:fluxes_props_text}. We can show that there is no
candidate with positive definite principal minor
$s_i\bar H^{ij}s_j=s_i\bar H^{(ij)}s_j$ (where $s$ is any spatial vector).
That means the four degrees of freedom in $\mathcal G_S$
are not sufficient to construct a positive symmetrizer for the principal
symbol. Thus, since all principal minors of a positive definite
matrix are positive as well, there is no symmetrizer which can be
written in terms of the metric.

We then ask for the candidate symmetrizers in the most general
set of matrices, i.e. the set which is only restricted by the
appropriate block structure \eqref{eq:partitioned_symmetrizer}.
There we find a 76 dimensional space of candidates,
$\tilde{\mathcal G}=\tilde{\mathcal G}_S\oplus\tilde{\mathcal G}_A$, but again
$\tilde{\mathcal G}_S$ is four dimensional. The most general
set of matrices of course contains all matrices that can be written
in terms of the metric, i.e. $\tilde{\mathcal G}_S = \mathcal G_S$.
Hence, $\tilde{\mathcal G}$ does not contain a matrix with positive
definite principal minor $s_i\bar H^{ij}s_j$, and therefore there is
no positive definite symmetrizer for the formulations
\eqref{eq:non_sym_hyp_form}.

This statement holds for a generic choice of the parameters in
\eqref{eq:non_sym_hyp_form}. To deal with special cases we apply
the rank criterion. It tells us that within
\eqref{eq:non_sym_hyp_form} one cannot find symmetrizers, except when 
$\mathcal G_S$ is at least five dimensional. This is the case
for $C_5=0$ only, which results in the (symmetric hyperbolic)
one parameter family $\mathcal F_3$.

Thus, the obvious conjecture is that formulations where the 
dimension of $\mathcal G_S$ is at least five are symmetric
hyperbolic, and that formulations where $\mathcal G_S$ is at most 
four dimensional are not.

To prove this one would need to analyze
positivity of candidates for three and four parameter families.
A problem that we could not solve so far.

We find that the dimension of $\mathcal G_S$ is at least five e.g. when 
one imposes $C_5=0$ in the family $\mathcal F_1$ \eqref{eq:scalar_four_param}, also for 
$C_7\neq 0$. But, as discussed above, for this three parameter family
we can only show that certain sub blocks of a candidate are positive definite.

\paragraph*{Discussion:} We have also constructed the characteristic 
variables for the four~$\mathcal{F}_1$ and three parameter~$\mathcal{F}_2$ 
families of strongly hyperbolic Hamiltonian formulations of GR. It is 
interesting that in every fully second order strongly hyperbolic 
Hamiltonian formulation (Maxwell, pure gauge and GR) we find that the 
fully second order characteristic variables always take the simple 
form 
\begin{align} 
U_i=\sum_{j}a_{ij}(\p_t\pm v\p_s)u_j,\label{eq:FSO_CharVarForm}
\end{align} 
with speeds $\pm v$ given constants $a_{ij}$ and primitive variables 
$q_i$. There are strongly hyperbolic fully second order systems without 
this property, for example the Z4 formulation coupled to the puncture 
gauge~\cite{Bernuzzi:2009ex}. At the moment it is not clear what causes 
the special form \eqref{eq:FSO_CharVarForm}.

We originally aimed to find strongly and symmetric hyperbolic Hamiltonian  
formulations with popular gauge conditions. Having analyzed a large class
of gauges, we are now in a position to return to that question. The 
generalized harmonic gauge is straightforward and as discussed previously 
is recovered with $C_i=0$, $i=1,\ldots,7$. More interesting is the puncture 
gauge~\eqref{eq:puncture_gauge}.

The formulation parameters in the strongly and symmetric hyperbolic formulations
can be used to adjust the characteristic speeds. If we regard them as scalar
functions of the position variables $(\gamma_{ij},\alpha,\beta^i)$ then the
equations of motion will contain derivatives of the parameters,
but in the principal part only the functions appear. Hence,
hyperbolicity of the formulations is not altered.

Here we want to present the evolution equations for lapse and shift for a
formulation that is very close to the puncture gauge. It belongs to the
symmetric hyperbolic family \eqref{eq:sym_hyp_form}. We choose
\begin{align}
\label{eq:sym_hyp_puncture}
\lambda_-^2 &= \mu_L,&
\lambda_+^2 &= \frac{4\mu_S\gamma^{1/3}}{3\alpha^2}-\frac13.
\end{align}
Thus, symmetric hyperbolicity requires
\begin{align}
\mu_L&>0, &
4 \gamma^{1/3}\mu_S &> \alpha^2,
\end{align}
which reduces to $0<\alpha<\sqrt{3}\gamma^{1/6}$ for the popular choice
$\mu_L=2/\alpha$, $\mu_S=3/4$. To make the comparison with other systems 
simpler we present the equations in terms of the \emph{Z4 variables}.
Instead of the canonical momenta $(\pi^{ij},\sigma,\rho_i)$
the Z4 formulation~\cite{Bona:2003fj} (which is not Hamiltonian) takes 
the extrinsic curvature and two fields called $(\Theta, Z_i)$ as evolved 
variables. The mapping to the canonical variables is defined through 
\eqref{eq:rel_pi_K} and
\begin{align}
\label{eq:trafo_Z4_canon}
\Theta &= \frac12 \frac{\alpha \sigma}{\sqrt{\gamma}},&
Z_i &= -\frac{\alpha \rho_i}{2\sqrt{\gamma}}.
\end{align}
In these variables, restricting to Brown's Hamiltonian results in a 
system that is as close to the Z4 formulation as possible. With the 
choice \eqref{eq:sym_hyp_puncture} we find
\begin{align}
\nonumber
\p_t\alpha &= \beta^i\p_i\alpha-\mu_L\alpha^2 K 
+\frac12\mu_L\alpha^2\Theta,\\
\nonumber
\p_t\beta^i &= \beta^j\p_j  \beta^i+
\mu_S\gamma^{1/3}\Gamma^i_{jk}\gamma^{jk}
+ \frac13(\mu_S\gamma^{1/3}-\alpha^2)\Gamma^k_{kj}\gamma^{ij}\\
&\qquad
+ 2\mu_S\gamma^{1/3}Z^i - \alpha D^i\alpha.
\end{align}
Hence, near the puncture (for $\alpha\rightarrow 0$)
we obtain the puncture gauge condition \eqref{eq:puncture_gauge}
when the constraints $\Theta=0$ and $Z_i=0$ are
satisfied.

Brown finds in \cite{Brown:2008cca} that a 
strongly hyperbolic Hamiltonian formulation with the puncture gauge 
does not exist. Here we have shown that with a slight modification of 
the $\Gamma$-driver shift condition even symmetric hyperbolicity can be
achieved. The new terms are small near the puncture.

The final measure of 
how useful the formulations with these gauges are will be 
determined by numerical experiments. Some of the ingredients for 
successful numerical evolutions are now understood. In this work we 
have taken care of questions related to the continuum formulation. 
For long-term stable numerical evolutions of puncture data a 
convenient choice of variables is an important feature as 
well~\cite{Witek:2010}. Furthermore, even if convenient 
variables can be taken on a suitably mathematically well-founded 
formulation, there is no guarantee that numerical simulations 
will be stable, either around flat-space or the puncture data 
we are interested in evolving. Numerical tests will be 
discussed in detail elsewhere.


\section{Conclusion}
\label{section:Conclusion}


Motivated partially by the success of symplectic integrators 
in many numerical applications we have studied Hamiltonian 
formulations of GR with live gauges. An essential property for 
any numerical application is that the underlying PDE problem 
is well-posed. Thus we have concentrated on well-posedness of 
the I(B)VP and considered the level of hyperbolicity of the 
equations of motion when the gauge choice is made
\emph{at the Hamiltonian level}. 

Several tools were developed to perform our analysis. We examined 
the relationship between the Hamiltonian energy and that required 
for symmetric hyperbolicity. In general we find that the 
Hamiltonian guarantees the existence of a candidate symmetrizer, 
which is typically not positive definite in applications. However 
there are important special cases in which Hamiltonian structure 
significantly simplifies hyperbolicity analysis. We have also 
translated the analysis in the literature on first order in 
time, second order in space systems to fully second order 
systems and now view the fully second order characteristic 
variables as the natural form whenever they can be constructed.

In our first applications we demonstrated that every strongly 
hyperbolic Hamiltonian formulation of electromagnetism is 
symmetric hyperbolic. We furthermore show that a generalization
of electromagnetism, with Hamiltonian structure, the pure gauge
system, is also always symmetric hyperbolic whenever it is 
strongly hyperbolic. 

There are several interesting consequences of the Hamiltonian 
approach to gauge choice. Normally from the free-evolution 
PDEs point of view the gauge choice of GR (or electromagnetism)  
allows one to choose equations of motion for some quantities 
and the freedom to add combinations of the constraints 
to the equations of motion. But in the Hamiltonian approach in 
GR for example the constraint addition is determined completely 
by the choice of the lapse and shift. This restricts the choice 
of formulation. There are choices of evolution equations for 
the lapse and shift that allow for a strongly or symmetric 
hyperbolic system without Hamiltonian structure that are forbidden 
when Hamiltonian structure is imposed, including the popular 
puncture gauge. We find in every example that the fully second 
order characteristic variables of the Hamiltonian systems always 
have a special form. We do not know how generally this property 
holds, but there are certainly examples of strongly hyperbolic 
formulations without Hamiltonian structure that do not exhibit 
it.

In our analysis of GR we find several families of strongly 
hyperbolic formulations and are able to choose the gauge 
parameters so that we get very close to the puncture 
gauge choice. Analysis of symmetric hyperbolicity is rather 
more difficult, but for every strongly hyperbolic formulation 
we are at least able to obtain partial results. We analyze 
several cases completely and find a Hamiltonian formulation 
that is symmetric hyperbolic with a small modification of the 
puncture gauge.

The next practical step for applications is to investigate 
whether or not any of the formulations constructed here may be 
used successfully in numerical evolutions, either with or 
without symplectic integration. There are however delicate 
issues involved in the choice of evolved variables.

An interesting question about GR is; what are the set of gauge 
choices that admit a well-posed I(B)VP? This question may be 
asked whether or not one restricts to formulations with 
Hamiltonian structure, and has still not been answered to our 
satisfaction in the general case. 


\acknowledgments


It is a pleasure to thank Carsten Gundlach, Jos\'e-Mar\'ia 
Mart\'in-Garc\'ia and Andrea Nerozzi for valuable discussions and 
comments on the manuscript. We also wish to thank Jos\'e-Mar\'ia 
Mart\'in-Garc\'ia for making several of his mathematica scripts 
available to us and Jerrold Marsden for his hints about previous
work. This work was supported in part by DFG grant 
SFB/Transregio~7 ``Gravitational Wave Astronomy''.


\appendix



\section{Hyperbolicity of fully 
second order systems}
\label{App:Hyperbolicity}


In this section we discuss the various
definitions of hyperbolicity for fully second order 
systems as summarized in Sect.~\ref{section:Definitions}.

The key in the proofs is the following. If two matrices
\begin{align}
\label{eq:fully2nd_from_FOITSOIS}
{\mathcal A}^{ij} &=  M^{-1}[V^{ij}-G^{i} M F^{j}],\nonumber\\
\mathcal B^i &= [F^iM^{-1}+M^{-1}G^i]M
\end{align}
are given then the matrix
\begin{align}
\mathfrak A^p{}_i{}^j &= \left(
\begin{array}{cc}
0 & \delta^p_i\\
{\mathcal A}^{pj} & \mathcal B^p
\end{array}
\right)
\end{align}
is similar to the matrix $A^p{}_i{}^j$:
\begin{align}
\label{eq:trafo_fully_2nd_FOITSOIS}
\mathfrak A^p{}_i{}^j &= \mathcal T^{-1}{}_i{}^k A^p{}_k{}^l \mathcal T{}_l{}^j,
\end{align}
where
\begin{align}
\label{eq:app_A_PP_mat}
A^p{}_i{}^j &=
\left(
\begin{array}{cc}
F^j\delta^p_i & M^{-1}\delta^p_i\\
V^{pj} & G^p
\end{array}
\right),\\
\mathcal T^{-1}{}_i{}^k &=
\left(
\begin{array}{cc}
\delta^k_i & 0\\
F^k & M^{-1}
\end{array}
\right),\\
\mathcal T_l{}^j &=
\left(
\begin{array}{cc}
\delta^j_l & 0\\
-MF^j & M
\end{array}
\right).
\end{align}

This shows immediately that a fully second order system is
strongly hyperbolic if and only if the fully second order
principal symbol \eqref{eq:fully_2nd_P_symbol} has a complete set
of eigenvectors with real eigenvalues.

Moreover one may easily prove that the notions of strong and symmetric
hyperbolicity for a first order in time, second order in space
system and an equivalent fully second order system agree.

Now we show that symmetric hyperbolicity of the fully second order
system is equivalent to the existence of a positive conserved quantity.

\begin{proof}
\emph{Fully second order system is symmetric hyperbolic $\Rightarrow$
a positive conserved quantity exists:}
Without loss of generality we assume the following reduction to first order
\begin{align}
\p_tq &= M^{-1}w + F^{i}\p_iq+S_q,\label{eq:App_FOITSOIS_1}\\
\nonumber
\p_tw &= V^{ij}\p_i\p_jq+G^i\p_iw+S_w,
\end{align}
which results in a first order in time, second order in space system
with principal part matrix \eqref{eq:app_A_PP_mat}.
We notice that in the principal part one can write
\begin{align}
\nonumber
\p_t\left(
\begin{array}{c}
\p_i q\\
\p_t q
\end{array}
\right)
&=
\mathcal T^{-1}{}_i{}^j\p_t\left(
\begin{array}{c}
\p_j q\\
w
\end{array}
\right)
=
\mathcal T^{-1}{}_i{}^jA^p{}_j{}^k\p_p\left(
\begin{array}{c}
\p_j q\\
w
\end{array}
\right)\\
&=
\mathfrak{A}^p{}_i{}^l
\p_p
\left(
\begin{array}{c}
\p_l q\\
\p_t q
\end{array}
\right).
\end{align}
It is then straightforward to check that if $H^{ij}$ is
a symmetrizer of the first order in time second order in space
system then $\mathcal T^\dag{}^i{}_k H^{kl} \mathcal T_l{}^j$ defines
a positive conserved quantity for the fully second order system.
Positivity follows from Sylvester's law of inertia. To show
conservation of the corresponding energy we notice that
\begin{align}
\label{eq:conserv_prop_2nd}
\nonumber
S_i\mathcal T^\dag{}^i{}_k &H^{kl}  \mathcal T_l{}^j \mathfrak A^p{}_j{}^m s_p S_m
=S_i\mathcal T^\dag{}^i{}_k H^{kl} A^p{}_l{}^j s_p \mathcal T_j{}^m S_m\\
\nonumber
&=\mathcal T^\dag{}^s{}_s S_k H^{kl} A^p{}_l{}^j s_p S_j\mathcal T_s{}^s\\
\nonumber
&=\mathcal T^\dag{}^s{}_s (S_k H^{kl} A^p{}_l{}^j s_p S_j)^\dag\mathcal T_s{}^s\\
\nonumber
&=(S_i \mathcal T^\dag{}^i{}_k H^{kl} A^p{}_l{}^j s_p \mathcal T_j{}^m S_m)^\dag\\
&=(S_i\mathcal T^\dag{}^i{}_k H^{kl} \mathcal T_l{}^j\mathfrak A^p{}_j{}^m s_p S_m)^\dag,
\end{align}
where $s_i$ is an arbitrary spatial unit vector and
\begin{align}
\label{eq:def_Si}
S_i &= \left(
\begin{array}{cc}
s_i & 0\\
0 & 1
\end{array}
\right).
\end{align}
Using the techniques presented in \cite{Gundlach:2005ta}
one can then show that \eqref{eq:conserv_prop_2nd} implies
conservation of the energy
\begin{align}
\epsilon &= \frac12
\left(\begin{array}{cc}
\p_jq^\dag & \p_tq^\dag
\end{array}\right)
\mathcal T^\dag{}^i{}_k H^{kl} \mathcal T_l{}^j
\left(\begin{array}{c}
\p_iq \\ \p_tq
\end{array}\right).
\end{align}

{\it A positive conserved quantity exists
$\Rightarrow$ Fully second order system is symmetric hyperbolic:}
This direction is obvious if one chooses the natural first order
reduction with $\p_t q=w$. The symmetrizer of the
first order in time second order in space system is then the
same as the symmetrizer of the fully second order system.
\end{proof}

\newcounter{dummy}
\setcounter{dummy}{\thelem}
\setcounter{lem}{\therepeatedLemma}
\begin{lem}
\label{lem:sym_hym_2nd}
Symmetric hyperbolicity of the fully second order 
system~\eqref{eq:evolution_second} is equivalent to the existence 
of {\it a second order symmetrizer and fluxes $(H_1,\phi^i,\phi^{ij}$)} satisfying
\begin{align}
\label{eq:fluxes_props}
\phi^{\dag i}&=\phi^{i},&\phi^{ij}&=\phi^{[ij]}=\phi^{ji\,\dag},
\end{align}
\begin{align}
\nonumber
s_is_js_k(\phi^i - \mathcal{B}^{\dag i}H_1)&\mathcal{A}^{jk}=\\
&=\mathcal{A}^{\dag jk}(\phi^i-H_1\mathcal{B}^{i})s_is_js_k,
\label{eq:2nd_Cons_first}
\end{align}
for every spatial vector $s_i$,  and
\begin{align}
\left(\begin{array}{cc}
H_1\mathcal{A}^{ij}+\mathcal{B}^{(i\,\dag}H_1\mathcal{B}^{j)}-
\mathcal{B}^{(i\,\dag}\phi^{j)} + \phi^{ij\,\dag} & \phi^i - \mathcal{B}^{\dag i}H_1\\  
\phi^j - H_1\mathcal{B}^{j} & H_1
\end{array}\right)\label{eq:2nd_Positivity}
\end{align}
hermitian positive definite.
\end{lem}
\setcounter{lem}{\thedummy}

\begin{proof}
We proceed by demonstrating equivalence
of \eqref{eq:fluxes_props}-\eqref{eq:2nd_Positivity} and the existence of 
a conserved positive quantity on the system.

{\it \eqref{eq:fluxes_props}-\eqref{eq:2nd_Positivity} $\Rightarrow$ conserved quantity exists:}
Consider the energy
\begin{align}
\label{eq:App_second_conserved}
E &= \int d^3 x\,\epsilon,\\
\nonumber
\epsilon&=\frac12
\left(\begin{array}{cc}
\p_iq^\dag & \p_t q^\dag
\end{array}\right)
H^{ij}
\left(\begin{array}{c}
\p_jq \\ \p_t q
\end{array}\right),
\end{align}
where $H^{ij}$ is the matrix~\eqref{eq:2nd_Positivity}.
The positivity of $H^{ij}$ obviously implies positivity of $E$. It remains to show that
$E$ is conserved.

Computing a time derivative of the energy density  $\epsilon$ we get
\begin{align}
2\p_t\epsilon &= \p_p\left(\p_tq^\dag\phi^p\p_tq\right)\\
\nonumber
&\qquad
+\p_i\left(q^\dag H_3^{[pi]\,\dag}\p_p\p_tq + \p_t\p_p q^\dag H_3^{[pi]}q\right)\\
\nonumber &\qquad
+\p_p\left(\p_iq^\dag\mathcal A^{ip\,\dag}H_1\p_tq + \p_tq^\dag H_1\mathcal A^{ip}\p_iq\right)\\
\nonumber &\qquad
+\p_i\p_p q^\dag\mathcal A^{ip\,\dag}H_2^{k\,\dag}\p_kq
+\p_i q^\dag H_2^i\mathcal A^{pk}\p_k\p_pq,
\end{align}
where $H_2^{i} = \phi^i-H_1\mathcal B^i$,
$H_3^{[pi]} = \phi^{pi\,\dag} + \mathcal B^{[p\,\dag}\phi^{i]}-
\mathcal B^{[p\,\dag}H_1\mathcal B^{i]}$.
The last two terms can be written as a divergence, too:
\begin{align}
\nonumber
&\p_i\p_p q^\dag\mathcal A^{ip\,\dag}H_2^{k\,\dag}\p_kq
+\p_i q^\dag H_2^i\mathcal A^{pk}\p_k\p_pq\\
\nonumber
&=
\p_i\p_p q^\dag\mathcal A^{(ip\,\dag}H_2^{k)\,\dag}\p_kq
+\p_i q^\dag H_2^{(i}{{\mathcal A}}^{pk)}\p_k\p_pq\\
\nonumber
&\qquad
+\frac23\p_i\p_p q^\dag(\mathcal A^{ip\,\dag}H_2^{k\,\dag}-\mathcal A^{k(p\,\dag}H_2^{i)\,\dag})\p_kq\\
\nonumber
&\qquad
+\frac23\p_i q^\dag (H_2^{i}\mathcal A^{pk}-H_2^{(k} {\mathcal A^{p)i}})\p_k\p_pq
\\
\nonumber
&=
\p_p\left(\p_i q^\dag H_2^{(i}{\mathcal A^{pk)}}\p_kq\right)\\
\nonumber
&\qquad
+\frac23\p_i\left(\p_p q^\dag(\mathcal A^{ip\,\dag}H_2^{k\,\dag}-\mathcal A^{kp\,\dag}H_2^{i\,\dag})\p_kq\right)\\
&\qquad
+\frac23\p_k\left(\p_i q^\dag (H_2^{i}\mathcal A^{pk}-H_2^{(k} {\mathcal A}^{p)i})\p_pq\right),
\end{align}
where the last equality holds because of \eqref{eq:2nd_Cons_first}
(for a tensor $T^{ijk}$ the equation $T^{(ijk)}=0$ is equivalent to the requirement
$T^{ijk}s_is_js_k=0$ for every covector $s_i$).
Hence $E$ is a conserved quantity.

{\it conserved quantity exists $\Rightarrow$ \eqref{eq:fluxes_props}-\eqref{eq:2nd_Positivity}:}
Assume that we have
the conserved quantity~\eqref{eq:App_second_conserved} with
\begin{align}
\label{eq:symmetrizer_block_form}
H^{ij} =
\left(
\begin{array}{cc}
H_3^{ij} & H_2{}^{i}\\
H_2{}^{i\,\dag} & H_1
\end{array}
\right),
\end{align}
and parametrize 
the fluxes of the energy by
\begin{align}
\Phi^i=\left(\begin{array}{cc}
\phi_3{}^{ijk} & \phi_2{}^{ij} \\  \phi_2{}^{\dagger ik} & 
\phi_1{}^i
\end{array}\right).
\end{align}
with $\phi_1^{\dagger i}=\phi_1^{i}$ and $\phi_3{}^{ijk\,\dag}=
\phi_3{}^{ikj}$. Computing and comparing $\p_t\epsilon$
and  $\p_i\Phi^i$ reveals that energy conservation implies
\begin{subequations}
\label{eq:energy_conv}
\begin{align}
\label{eq:energy_conv_1}
2\phi_3^{(jk)i} = \mathcal{A}^{jk\,\dag} H_2{}^{i\,\dag},\\
\label{eq:energy_conv_2}
2\phi_2^{(ij)}=\mathcal{A}^{ij\,\dag}H_1,\\
\label{eq:energy_conv_3}
2\phi_2^{ij} = H_3{}^{ji} + H_2{}^j\mathcal{B}^i,\\
\label{eq:energy_conv_4}
2\phi_1^i = \mathcal{B}^{i\,\dag}H_1+H_2{}^{i}
\end{align}
\end{subequations}
The solution of \eqref{eq:energy_conv_4} is obvious:
\begin{align}
H_2{}^i=2\phi_1^i-\mathcal B^{i\,\dag}H_1.
\end{align}
From \eqref{eq:energy_conv_2} and \eqref{eq:energy_conv_3} we get
\begin{align}
\nonumber
H_3^{ij} &= 2\phi_2^{ji} - H_2{}^i\mathcal B^j\\
\nonumber
&= \mathcal A^{ij\,\dag}H_1 - 2\phi_2^{[ij]} - H_2{}^i\mathcal B^j\\
&= \mathcal A^{ij\,\dag}H_1 - 2\phi_2^{[ij]} - 2\phi_1^i\mathcal B^j + \mathcal B^{i\,\dag} H_1\mathcal B^j.
\end{align}
Since $H^{ij}$ defines a symmetrizer we also have $H_3^{ij\,\dag} = H_3^{ji}$.
When we identify
\begin{align}
\phi^i&=2\phi_1^i,\\
\nonumber
\phi^{ij}&=2\phi_2^{[ij]}+2\phi_1^{[i}\mathcal B^{j]}-\mathcal B^{[i}H_1\mathcal B^{j]}
\end{align}
this implies that \eqref{eq:2nd_Positivity} is hermitian,
and \eqref{eq:fluxes_props} is obvious.

Concerning \eqref{eq:energy_conv_1} it was proven in \cite[section IVB]{Gundlach:2005ta}
that this equation together with $\phi^{ijk\,\dag} = \phi^{ikj}$
has a solution if and only if
\begin{align}
\mathcal{A}^{(jk\,\dag} H_2{}^{i)\,\dag} = H_2{}^{(i}\mathcal{A}^{jk)}.
\end{align}
This equation is equivalent to \eqref{eq:2nd_Cons_first}.

Finally, the positivity of the symmetrizer $H^{ij}$ implies that the
matrix \eqref{eq:2nd_Positivity} is positive definite.
\end{proof}


\section{Proof of the rank criterion}
\label{app:rank_criterion}


In section \ref{sec:rank_criterion} we described the rank criterion,
a feature that can be used as a necessary condition for symmetric
hyperbolicity. Here we discuss the proof of this criterion.

We follow the notation of section \ref{sec:rank_criterion}. Thus,
$\mathcal C^l\subset\mathbb R^l$ is a set of formulation parameters, $c$,
and we consider $k\times k$ principal part matrices $A^p{}_i{}^j(c)$
that depend continuously on $c$. We search for symmetrizers of
$A^p{}_i{}^j(c)$ in the set of hermitian $k\times k$-matrices, $\mathcal G$,
that is the image of a linear map
$G^{ij}:\mathbb R^n\rightarrow\mathbb R^{k\times k}$, $g\mapsto G^{ij}(g)$.
The requirement that $G^{ij}(g)$ is a candidate symmetrizer for the
formulation $c$ defines linear equations on $g$ of the form
\begin{align}
B(c)g = 0,
\end{align}
where $B(c)\in\mathbb R^{m\times n}$ depends continuously on $c$.

We prove the following
\begin{lem}
Let $\bar{\mathcal C}^l\subset \mathcal C^l$ be the set of formulations
where the rank of $B(\bar c)$ is $N$ and that do not possess symmetrizers
in $\mathcal G$.
Let further $c$ be a formulation such that the intersection of every
neighborhood $U_c\subset\mathcal C^l$ of $c$ with $\bar{\mathcal C}^l$ is
not empty, $U_c\cap\bar{\mathcal C}^l\neq\emptyset$.

If there is a symmetrizer for $c$ in $\mathcal G$ then the rank of
$B(c)$ is smaller than $N$.
\end{lem}

For the proof we need the following
\begin{prop}
Let $B, \tilde B\in\mathbb{R}^{m\times n}$ with $\|B-\tilde B\|<\varepsilon$
and $g\in ker B$.
Let further be
\begin{align}
K(\tilde B)=\inf_{v\in(ker \tilde B)^\bot\setminus\{0\}}\frac{\|\tilde B v\|}{\|v\|},
\end{align}
where $V^\bot:=\{v\in\mathbb{R}^n:\langle v,\bar v\rangle=0,\;\forall\bar v\in V\}$.

Then we get
\begin{align}
\inf_{h\in ker\tilde B}\|g-h\|=\min_{h\in ker\tilde B}\|g-h\|<\frac{\varepsilon}{K\big(\tilde B\big)}\|g\|.
\end{align}
\end{prop}

\begin{proof}
It is clear that there exists a $\tilde g\in ker\tilde B$ such that
$\inf_{h\in ker\tilde B}\|g-h\|=\|g-\tilde g\|$. It satisfies
$g-\tilde g\in(ker \tilde B)^\bot$. If $g=\tilde g$ nothing needs to
be shown. Hence, we assume that $g-\tilde g\in(ker \tilde B)^\bot\setminus\{0\}$.
The definition of $K$ then implies
\begin{align}
\nonumber
K(\tilde B)\|g-\tilde g\|&\leq\|\tilde B(g-\tilde g)\|=\|\tilde Bg\|
=\|(\tilde B-B)g\|\\
&\leq\|\tilde B-B\|\,\|g\|
<\varepsilon\|g\|.
\end{align}
\end{proof}

Now we prove the above lemma.
Since $B:\mathcal C^l\rightarrow\mathbb{R}^{m\times n}$ is continuous,
there exists for every $\varepsilon>0$ a $\delta>0$ such that
$\|B(c)-B(\tilde c)\|<\varepsilon$ for every $\tilde c$ with $\|c-\tilde c\|<\delta$.
We assume that $g\in\mathbb{R}^n$ satisfies $B(c)g=0$. Then, according to the proposition
there is for every $\tilde c$ with $\|c-\tilde c\|<\delta$ a
$\tilde g\in ker(B(\tilde c))$ that satisfies
\begin{align}
\|g-\tilde g\|<\frac{\varepsilon}{K(B(\tilde c))}\|g\|.
\end{align}

Since $U_c\cap\bar {\mathcal C}^l\neq\emptyset$ for every neighborhood
$U_c$ of $c$ we can choose $\bar c\in\bar{\mathcal C}^l$ such that
$\|c-\bar c\|<\delta$. Then we get for every $\bar g\in ker B(\bar c)$ that
$G^{ij}(\bar g)$ is not positive definite (because otherwise there is a symmetrizer for
$\bar c\in\bar{\mathcal C}^l$).

Hence, for every $\varepsilon>0$ and every candidate $G^{ij}(g)$ of $c$
we find a candidate $G^{ij}(\bar g)$ of $\bar c$ that is not positive
definite such that
\begin{align}
\|g-\bar g\|<\frac{\varepsilon}{K(B(\bar c))}\|g\|.
\end{align}

On the other hand, since we find a symmetrizer for $c$ in $\mathcal G$
there exists a point $g_+\in ker B(c)$ such that $G^{ij}(g_+)$
is positive definite.

We notice that $G^{ij}$ is a continuous map (because it is linear) and that the map which
assigns to a matrix its smallest eigenvalue is continuous, too. This
implies that there is a neighborhood $U_{+}\subset\mathcal G$ of $g_+$ such that
the matrix $G^{ij}(h)$ is positive definite for every $h\in U_{+}$.

We choose $\bar g_+\in ker(B(\bar c))$ such that
\begin{align}
\label{eq:gmbarg_estimate}
\|g_+-\bar g_+\|<\frac{\varepsilon}{K(B(\bar c))}\|g_+\|.
\end{align}
We know that $\bar g_+$ is not in $U_{+}$.
Therefore there exists a $\rho>0$ such that $\|g_+-\bar g_+\|\geq\rho$.

Equation \eqref{eq:gmbarg_estimate} then implies that for every $\varepsilon>0$ there
is a formulation $\bar c\in\bar{\mathcal C}^l$ with
\begin{align}
K(B(\bar c))<\Lambda\varepsilon,
\end{align}
with the constant $\Lambda=\|g_+\|/\rho$.

To show that this implies a jump in the rank of $B$ we consider a
sequence $\{\bar c_n\}_{n\in\mathbb N}\subset\bar {\mathcal C}^l$ such that
\begin{align}
\lim_{n\rightarrow\infty}\bar c_n=c.
\end{align}
We notice that $\lim_{n\rightarrow\infty}K(B(\bar c_n)) = 0$.

Then, for every $\bar c_n$ there exists a $\bar g_n\in (ker B(\bar c_n))^\bot$ such that
$\|\bar g_n\|=1$ and $K(B(\bar c_n))=\|B(\bar c_n)\bar g_n\|$ (the set
$\{\|\bar g\|=1\}$ is compact).
From the sequence $\{\bar g_n\}$ we choose a convergent subsequence $\{\bar h_m\}$
and get
\begin{align}
\nonumber
0
&=\lim_{m\rightarrow\infty}K(B(\bar c_m))\bar h_m
=\lim_{m\rightarrow\infty}B(\bar c_m)\bar h_m\\
&=\left(\lim_{m\rightarrow\infty}B(\bar c_m)\right)\left(\lim_{m\rightarrow\infty}\bar h_m\right)
=B(c)\bar h,
\end{align}
with $\|\bar h\|=1$. Analogously one can show that the limit of every convergent
sequence $\{g_n\in ker B(\bar c_n)\}$ is in the kernel of $B(c)$.

It follows that the dimension
of the kernel of $B(c)$ is bigger than the dimension of the kernel of
$B(\bar c_n)$, i.e. the rank of $B(c)$ is smaller than $N$.
\qed


\section{Derivation of conditions for strong hyperbolicity}
\label{app:strong_hyp_cond}


In this appendix we discuss how conditions for strong
hyperbolicity can be derived on the Hamiltonian formulations
of GR introduced in section \ref{sec:Hamiltonian_formulations}.
One finds that the special structure of the principal symbol
which we assume in this article is very helpful for this
purpose. We consider first 
order in time, second order in space formulations, but most 
steps can be directly carried forward to the fully second order 
case.

The key will be to write the symbol in a simple standard form
where the conditions can be read off easily. This is meant to 
be a preliminary step for the construction of positive definite 
symmetrizers in the next appendix \ref{app:symmetrizer_constr}.

We are interested in Hamiltonian formulations of GR with the Hamiltonian
$\mathcal H$ given in \eqref{eq:general_Hamiltonian}. Thus, we are free
to choose the seven formulation parameters $C_1,\ldots,C_7$. Yet,
we replace $C_1$, $C_2$ and $C_3$ by $\lambda_-^2$, $\lambda_V^2$ and
$\lambda_+^2$  respectively, using \eqref{eq:defn_lam_pm} and
\eqref{eq:tensor_vector_speeds}.
The given matrix expressions are valid in the Z4 variables
\eqref{eq:trafo_Z4_canon}. But the analysis in the canonical variables is analogous.
The only difference are factors of 2 at some places.

As discussed in section \ref{sec:Hamiltonian_formulations} for the hyperbolicity
analysis one can assume without loss of generality $\alpha=1$, $\beta^i=0$. I.e. the
principal part matrix depends on the 3-metric and the formulation parameters only.

After changing the order of variables to $(\gamma_{ij},\alpha,Z_i,K_{ij},\Theta,\beta^i)$
we find that the principal symbol, $P^s$, has the block structure
\begin{align}
\label{eq:XY_defn}
P^s &= \left(
\begin{array}{cc}
0 & X\\
Y & 0
\end{array}
\right).
\end{align}

To this matrix we apply similarity transformations to bring it to a form where
the conditions for strong hyperbolicity can be read off easily. First we consider
the case where $X$ is invertible. For our example that means $\lambda_-^2\neq 0$.
We get
\begin{align}
\tilde P^s &= T_X^{-1}P^s T_X =
\left(
\begin{array}{cc}
0 & \mathbf 1\\
X Y & 0
\end{array}
\right),
\end{align}
where
\begin{align}
T_X &= \left(
\begin{array}{cc}
\mathbf 1 & 0\\
0 & X^{-1}
\end{array}
\right).
\end{align}

Since we assumed that the principal part matrix can be written in terms of the
metric only one can write $XY$ in block diagonal form.
We define the orthogonal projector, $q$, to the direction $s$ and decompose
the vectors and symmetric 2-tensors as
\begin{align}
\nonumber
V^i &= q^i_A V^A + s^i V^s,\\
\nonumber
T^{ij} &= \left(q^i{}_{(A} q^j{}_{B)}-\frac12 q^{ij}q_{AB}\right) T_{TF}^{AB}
+ \sqrt{2} q_A{}^{(i}s^{j)} \hat T^{sA}\\
&\qquad+ \frac1{\sqrt{2}} q^{ij} \hat T^{qq}
+ s^i s^j T^{ss},
\end{align}
where we use $\hat T^{sA} := s_i q_j{}^A T^{ij}/\sqrt{2} = T^{sA}/\sqrt{2}$
and $\hat T^{qq} := \sqrt{2} T^{ij}q_{ij} = \sqrt{2} T^{qq}$ instead of the
usual decomposition with $T^{sA}$ and $T^{qq}$, because it makes the
transformation map orthogonal.
We apply this decomposition to $\tilde P^s$ and get that it
decomposes into a trace-free tensor, a vector and a scalar block,
$\tilde P^s_{TF}$, $\tilde P^s_{V}$ and $\tilde P^s_{S}$ respectively.
The corresponding submatrices of $XY$ are
\begin{align}
XY_{TF}\,{}_{ij}{}^{kl} &= q_i{}^{(l}q_j{}^{k)}-\frac12 q_{ij}q^{kl},\\
XY_{V}\,{}_i{}^k &= q_i{}^k \left(
\begin{array}{cc}
\lambda_V^2 & -2 \sqrt{2} (\lambda_V^2 - 1 - 2 C_6)\\
0 & \lambda_V^2
\end{array}
\right),\\
XY_{S} &=\left(
\begin{array}{cc}
A & 0\\
B & C
\end{array}
\right),
\end{align}
where
\begin{align}
\nonumber
A &= \left(
\begin{array}{cc}
\lambda_+^2 - (1 + C_5) C_7 & C_7/\sqrt{2}\\
\nonumber
\sqrt{2}\left(\lambda_+^2 - (1+C_5)(C_7 + \lambda_-^2)\right) & C_7 + \lambda_-^2
\end{array}
\right),\\
\nonumber
C &= \left(
\begin{array}{cc}
\lambda_-^2 - C_5 C_7 & C_7\lambda_+^2/2\\
\nonumber
-2C_5 & \lambda_+^2
\end{array}
\right),\\
B &= \left(
\begin{array}{cc}
B_{11} &
\frac{\lambda_-^2 - 1 - 8C_4 + C_7(2 \lambda_+^2-4\lambda_V^2+1)}{2\sqrt{2}}\\
B_{12} &
\frac{C_7 + \lambda_-^2 + 2\lambda_+^2 - 4\lambda_V^2+1}{\sqrt{2}}
\end{array}
\right),\\
\nonumber
2B_{11}&=1 + C_5 + 8C_4(1 + C_5) + 8C_6C_7 - 4\lambda_-^2\\
\nonumber&\qquad - 4(\lambda_V^2-1)(C_7 + \lambda_-^2) + 3(C_7 + \lambda_-^2)\lambda_+^2,\\
\nonumber
B_{12}&= 8 C_6 -4 (\lambda_V^2-1) - (1 + C_5)(C_7 + \lambda_-^2) + \lambda_+^2.
\end{align}
We do not get the vanishing upper right block in $XY_S$
when we start from the fully second order system. Therefore
we prefer the first order in time system here.

We see that the trace-free tensor block is always diagonalizable,
because $q_{(i}{}^{k}q_{j)}{}^{l}-q_{ij}q^{kl}/2$ is the identity.
The vector block has the structure
\begin{align}
q_i{}^{k}\left(
\begin{array}{cccc}
0 & 0 & 1 & 0\\
0 & 0 & 0 & 1\\
\lambda_V^2 & b_1 & 0 & 0\\
0 & \lambda_V^2 & 0 & 0
\end{array}
\right).
\end{align}
It is hence diagonalizable with real eigenvalues
if and only if $b_1=0$, i.e. $C_6=(\lambda_V^2-1)/2$,
and $\lambda_V^2 > 0$.

Concerning the scalar part we use that the upper right block of
$XY_S$ vanishes. A necessary condition for its diagonalizability
is hence that $A$ and $C$ are diagonalizable. One finds that the
eigenvalues of $A$ and $C$ are the same, namely
\eqref{eq:Speeds_v+-}.
Hence, $A$ and $C$ are diagonalizable if
$(\lambda_+^2-\lambda_-^2)^2+C_5^2C_7^2
+2C_5C_7(\lambda_+^2+\lambda_-^2) \neq 0$.

When $A$ and $C$ have only one eigenvalue those matrices must be
(as $2\times 2$-matrices) already diagonal. This implies $C_5=0=C_7$.
Moreover, since there is only one eigenvalue, the matrix $B$ must vanish.
One obtains \eqref{eq:scalar_one_param_equal_speeds}.
The squared characteristic speeds are
\begin{align}
\lambda_V^2 &= 1/4(1+3\lambda_-^2),&
v_{+/-}^2 = \lambda_{+/-}^2 &= 1+2C_1,
\end{align}
and $P^s$ has real eigenvalues if $\lambda_V^2\geq 0$,
$\lambda_{-}^2\geq 0$.
However, the case $\lambda_{-}^2=0$ is not allowed here,
because $X$ would be singular. Hence, we get $\lambda_-^2>0$.

For $v_+^2\neq v_-^2$ we can transform to a basis where
$A$ and $C$ are diagonal. Let $T_A$ and
$T_C$ diagonalize $A$ and $C$ respectively. We apply the transformation
\begin{align}
\bar P^s_S &= \tilde T_S^{-1} \tilde P^s_S \tilde T_S,
\end{align}
with
\begin{align}
\tilde T_S=\mbox{diag}\left(
T_A , T_C, T_A, T_C
\right)
\end{align}
and get the scalar block $\bar P^s_S$ with the following structure
\begin{align}
\overline{XY}_{S} &=
\left(
\begin{array}{cccc}
v_+^2 & 0 & 0 & 0\\
0 & v_-^2 & 0 & 0\\
\bar B_{11} & \bar B_{12} & v_+^2 & 0\\
\bar B_{21} & \bar B_{22} & 0 & v_-^2
\end{array}
\right).
\end{align}
Hence, to make $\bar P^s_S$ diagonalizable one needs that
$\bar B_{11} = 0 = \bar B_{22}$. Those expressions are
however very long so we do not present them. We discuss the solutions 
of the equations $\bar B_{11} = 0 = \bar B_{22}$ in section 
\ref{sec:Hamiltonian_formulations}. There we find three families of 
strongly hyperbolic formulations that we denote $\mathcal F_{1/2/3}$.

In the next section we prove symmetric hyperbolicity
of the strongly hyperbolic formulations in the four parameter family
$\mathcal F_1$ \eqref{eq:scalar_four_param} which also satisfy $C_7=0=C_5$.
In that case we get
\begin{align}
\nonumber
\bar B_{11} &= (\lambda_-^2 - 1 - 8C_4)/\sqrt{4},\\
\bar B_{22} &= (3\lambda_+^2 - 8\lambda_V^2 + 8C_6 + 5)/\sqrt{2}.
\end{align}
Hence strong hyperbolicity implies
\begin{align}
\nonumber
\lambda_V^2 &= (3\lambda_+^2+1)/4,&
C_4 &= (\lambda_-^2-1)/8,\\
\nonumber
C_6 &= 3(\lambda_+^2-1)/8,\\
\lambda_-^2 &> 0,&
\lambda_+^2 &\geq 0.
\end{align}
For $\lambda_+^2=0$ the matrix $XY_S$ has vanishing eigenvalues. Hence,
$\tilde P_S^s$ is not diagonalizable. Thus, the case
$\lambda_+^2=0$ must be excluded.

\paragraph*{Singular $X$:}
In the analysis above we needed the assumption that $X$ is invertible.
If $X$ is not invertible ($\lambda_-^2=0$) then the diagonalizability of
$XY$ is only a necessary condition for the diagonalizability of $P^s$
(when $P^s$ is diagonalizable then also $(P^s)^2$, and hence $XY$ and
$YX$ are).

Indeed it turns out that the derived conditions are not sufficient in
that case. However, having a complete set of eigenvectors for $(P^s)^2$
one obtains restrictions on the formulation parameters through the
requirement that every eigenvector of $(P^s)^2$ must be an eigenvector
of $P^s$ (details can be found e.g. in \cite{Richter:2009ff}). One finds
a one parameter family of strongly hyperbolic formulations. But there
the matrix $M$ \eqref{eq:Def_Minv} is not invertible, i.e. we cannot transform
to a fully second order free evolution system. Therefore those
formulations are not of interest here.


\section{Construction of symmetrizers}
\label{app:symmetrizer_constr}


In this appendix a possible procedure for the construction of 
a positive symmetrizer for formulations of GR is discussed. The 
difficulty is usually not the construction of general candidate 
symmetrizers but the proof of positivity. We use a method that 
simplifies the latter task. The approach works as follows.

First the principal symbol is transformed to a simple standard 
form using the calculations of the previous appendix 
\ref{app:strong_hyp_cond}. It is then possible to construct the 
set of all positive definite matrices, $\mathcal G$, that symmetrize 
this transformed symbol. The next step is to make an ansatz 
for candidate symmetrizers and to compare the sub block of 
those \emph{ansatz candidates} with the matrices in $\mathcal G$.
If there is an overlap then the positivity of the matrices in 
$\mathcal G$ may be used to simplify the positivity conditions 
of the ansatz candidate.

It seems to be the most important task to find that kind of 
simplifications, because the reason why the proof of positivity 
is hard are the complicated expressions in the eigenvalues of the 
candidate symmetrizers. Starting from a symmetrizer of the symbol 
allows to combine some terms.

We will show that in this way it is possible to prove symmetric 
hyperbolicity for the two parameter family of strongly hyperbolic 
formulations \eqref{eq:sym_hyp_form} (it is the special case of the
four parameter family $\mathcal F_1$ \eqref{eq:scalar_four_param}
with $C_5 = 0 = C_7$).

In section \ref{sec:Hamiltonian_formulations} we presented more
strongly hyperbolic formulations, but we were not able to prove 
the existence of symmetrizers in the general case. The reason is 
that already the positivity conditions on the symmetrizers of the 
symbol, i.e. the matrices in $\mathcal G$, become complicated and 
we were not able to reduce them.

Yet, for another two parameter family of strongly hyperbolic 
formulations, \eqref{eq:non_sym_hyp_form}, we were able to show that
it is not symmetric hyperbolic, which means that in GR there 
are Hamiltonian formulations that are strongly but not symmetric 
hyperbolic.

\paragraph*{Decomposition of the ansatz candidate.}

In the construction of a symmetrizer we assume that we start from 
a reasonable ansatz candidate, $G$. The first step is a decomposition 
of its derivative indices into longitudinal and transverse part:
\begin{align}
\label{eq:deriv_decomp}
G^{ij} =
\left(
\begin{array}{cc}
s^i & q^i{}_A
\end{array}
\right)
\left(
\begin{array}{cc}
G^{ss} & G^{sB}\\
G^{As} & G^{AB}\\
\end{array}
\right)
\left(
\begin{array}{c}
s^j\\
q_B{}^j
\end{array}
\right),
\end{align}
with $G^{ss\,\dag}=G^{ss}$, $G^{As\,\dag}=G^{sA}$, $G^{AB\,\dag}=G^{BA}$.
We can apply the same decomposition to the principal part matrix:
\begin{align}
s_p A^{p}{}_i{}^j =
\left(
\begin{array}{cc}
s_i & q_i{}^A
\end{array}
\right)
\left(
\begin{array}{cc}
A^s{}_s{}^s & A^s{}_s{}^B\\
A^s{}_A{}^s & A^s{}_A{}^B
\end{array}
\right)
\left(
\begin{array}{c}
s^j\\
q_B{}^j
\end{array}
\right).
\end{align}
Due to the structure of $A^p{}_i{}^j$ we get
$s_p q^i_A A^p{}_i{}^j=0$, and it follows that $G^{ij}$
is a candidate symmetrizer of $A^p{}_j{}^k$ if and only if for
all spatial vectors $s$ the matrix $G^{ss}A^s{}_s{}^s$
is hermitian.

Later we discuss the construction of a positive definite $G^{ss}$.
But given this matrix it is clear that $G^{ij}$ is positive
definite if and only if
\begin{align}
\left(
\begin{array}{cc}
\mathbf 1 & 0\\
B^A & q^A{}_C
\end{array}
\right)&
\left(
\begin{array}{cc}
G^{ss} & G^{sC}\\
G^{Ds} & G^{CD}
\end{array}
\right)
\left(
\begin{array}{cc}
\mathbf 1 & B^{A\,\dag}\\
0 & q_D{}^{B}
\end{array}
\right)\\
\nonumber
&=
\left(
\begin{array}{cc}
G^{ss} & 0\\
0 & G^{AB}-G^{As}(G^{ss})^{-1} G^{sB}
\end{array}
\right)>0,
\end{align}
with
\begin{align}
B^A &=-G^{As}(G^{ss})^{-1}
\end{align}
is. Hence, given a positive definite $G^{ss}$ we need to prove
positivity of
\begin{align}
\label{eq:remaining_pos_mat}
G^{AB}-G^{As}(G^{ss})^{-1} G^{sB}.
\end{align}

In our formulations of GR that means for the analysis of
positivity the full candidate symmetrizer, a $40\times 40$ matrix,
is reduced to two $20\times 20$ matrices.

\paragraph*{Construction of a positive $G^{ss}$.}

Now we want to construct a matrix $G^{ss}$ that
symmetrizes the principal symbol $P^s = A^s{}_s{}^s$.
One cannot expect that every symmetrizer of $P^s$ is
part of a symmetrizer for the full principal part matrix, i.e. the
constructed $G^{ss}$ must be sufficiently general
to make it possible to adjust parameters in the
construction of the full symmetrizer later on.
But it must not contain too many parameters, because
this complicates the positivity analysis of the matrix
\eqref{eq:remaining_pos_mat}.

As a good compromise we found it reasonable to start from
the ansatz that the full symmetrizer depends on the metric only,
but is not restricted otherwise:
\begin{align}
\label{eq:ansatz_symmetrizer}
&G^{ij\,kl\,mn}=\\
\nonumber
&\left(
\begin{array}{cccccc}
G_{11}^{ij\,kl\,mn} & G_{12}^{ij\,kl} & 0 & 0 & 0 & G_{16}^{i\,kl\,m} \\
G_{12}^{ji\,mn} & G_{22}^{ij} & 0 & 0 & 0 & G_{26}^{i\,m} \\
0 & 0 & G_{33}^{ij}{}_{k\,m} & G_{34}^{i}{}_{k\,mn} & G_{35}^{i}{}_{k} & 0 \\
0 & 0 & G_{34}^{j}{}_{m\,kl} & G_{44\,kl\,mn} & G_{45\,kl} & 0 \\
0 & 0 & G_{35}^{j}{}_{m} & G_{45\,mn} & G_{55} & 0 \\
G_{16}^{j\,mn\,k} & G_{26}^{j\,k} & 0 & 0 & 0 & G_{66}^{k\,m}
\end{array}
\right),
\end{align}
with
{\allowdisplaybreaks
\begin{align*}
G_{11}^{ij\,kl\,mn} &=
G_{11}^1 \gamma^{ij} \gamma^{kl} \gamma^{mn} +
2 G_{11}^2 \gamma^{ij} \gamma^{k(m} \gamma^{n)l}\\
&\qquad +
2 G_{11}^3 (\gamma^{i(k} \gamma^{l)j} \gamma^{mn} +
            \gamma^{i(m} \gamma^{n)j} \gamma^{kl})\\
&\qquad +
2 G_{11}^4 (\gamma^{ik} \gamma^{j(m} \gamma^{n)l} +
            \gamma^{il} \gamma^{j(m} \gamma^{n)k}\\
&\qquad\qquad + \gamma^{im} \gamma^{j(k} \gamma^{l)n} +
            \gamma^{in} \gamma^{j(k} \gamma^{l)m})\\
&\qquad +
A_{11}^1 (\gamma^{k[i}\gamma^{j](m}\gamma^{n)l} +
          \gamma^{l[i}\gamma^{j](m}\gamma^{n)m}),\\
G_{12}^{ij\,kl} &=
2 G_{12}^1 \gamma^{i(k} \gamma^{l)j} +
  G_{12}^2 \gamma^{ij} \gamma^{kl},\\
G_{22}^{ij} &= G_{22}^1 \gamma^{ij},\\
G_{33}^{ij}{}_{k\,m} &=
2 G_{33}^1 \delta^{i}_{(k} \delta^{j}_{m)} +
G_{33}^2 \gamma^{ij} \gamma_{km}
+ A_{33}^1 \delta^{[i}_k\delta^{j]}_m,\\
G_{16}^{i\,kl\,m} &=
2 G_{16}^1 \gamma^{i(k} \gamma^{l)m} +
G_{16}^2 \gamma^{im} \gamma^{kl},\\
G_{26}^{i\,m} &= G_{26}^1 \gamma^{im},\\
G_{34}^{i}{}_{k\,mn} &=
G_{34}^1 \delta^i_{(m} \gamma_{n)k} +
G_{34}^2 \delta^i_k \gamma_{mn},\\
G_{35}^{i}{}_{k} &= G_{35}^1 \delta^{i}_k,\\
G_{44\,kl\,mn} &=
2 G_{44}^1 \gamma_{k(m} \gamma_{n)l} +
G_{44}^2 \gamma_{kl} \gamma_{mn},\\
G_{45\,kl} &= G_{45}^1 \gamma^{kl},\\
G_{66}^{k\,m} &= G_{66}^1 \gamma^{km}.
\end{align*}
}

Now, from the previous appendix \ref{app:strong_hyp_cond}
we know that for $\lambda_-^2\neq 0$
there exists a matrix $T$ and a diagonal matrix $\Lambda$ such that
\begin{align}
\tilde P^s &= T^{-1}P^s T =
\left(
\begin{array}{cc}
0 & \mathbf 1 \\
\Lambda & 0
\end{array}
\right).
\end{align}

We decompose the derivative indices of
the ansatz candidate according to \eqref{eq:deriv_decomp}
and apply a congruence transformation using the matrix $T$:
\begin{align}
\label{eq:trafo_cand_diag_basis}
&\left(
\begin{array}{cc}
\tilde G^{ss} & \tilde G^{sB}\\
\tilde G^{As} & G^{AB}
\end{array}
\right)=\\
\nonumber
&\qquad=
\left(
\begin{array}{cc}
T^\dag & 0\\
0 & q^A{}_C
\end{array}
\right)
\left(
\begin{array}{cc}
G^{ss} & G^{sD}\\
G^{Cs} & G^{CD}
\end{array}
\right)
\left(
\begin{array}{cc}
T & 0\\
0 & q_D{}^B
\end{array}
\right).
\end{align}
The condition that $G^{ij}$ is a candidate symmetrizer of the full problem
then becomes
\begin{align}
\label{eq:candidate_condition}
\tilde G^{ss}\tilde P^s = (\tilde G^{ss}\tilde P^s)^\dag.
\end{align}
This is a linear equation on the parameters $G_{ab}^c$ in $G$. It is guaranteed
that the solution does not depend on the direction of $s$, because both the
principal part matrix $A^p{}_i{}^j$ and the ansatz candidate $G^{ij}$ depend on
the metric only.

For the family \eqref{eq:sym_hyp_form}
the solutions of \eqref{eq:candidate_condition} are five dimensional. For the
positivity analysis it is important to note that that $\tilde G^{ss}$
decomposes into $2\times 2$ blocks.

To simplify expressions we introduce new parameters:
\begin{align}
s_{V} &:= 2 G_{11}^2 + \frac{G_{66}^1}{2(3\lambda_+^2+1)},\\
\nonumber
s_{44} &:= G_{33}^1 + \frac{G_{66}^1}{2(3\lambda_+^2+1)},\\
\nonumber
s_{22} &:= \frac{4G_{11}^2 - 2G_{33}^1 + G_{44}^2 - G_{66}^1(1 + 3\lambda_+^2)^{-1}}{\lambda_-^4},\\
\nonumber
s_{12} &:= \frac{4G_{11}^2 - 2G_{33}^1 + G_{44}^2 + 2G_{45}^1 - G_{66}^1(1 + 3\lambda_+^2)^{-1}}{2\lambda_-^2},
\end{align}
where $s_{V}$ appears in the vector and $s_{12}$, $s_{22}$ as well as
$s_{44}$ in the scalar block.

With those parameters and the observation that $\tilde G^{ss}$ decomposes into
$2\times 2$ blocks the positivity conditions on $\tilde G^{ss}$
can be reduced using computer algebra. One finds that for solutions of \eqref{eq:candidate_condition}
$\tilde G^{ss}$ is positive definite if and only if
\begin{align}
\nonumber
s_{44} &<
-\frac2{9\lambda_+^2s_{22}}\big(\lambda_-^2s_{12}^2\\
\nonumber&\qquad\qquad
+ s_{22}\left(G_{11}^2(4 + 9\lambda_+^2) - 2s_V(1 + 3\lambda_+^2)\right)\big)\\
s_{22} &> -\frac{\lambda_-^2 s_{12}^2}{2(1 + 3\lambda_+^2)(2G_{11}^2 - s_V)},
\end{align}
and
\begin{align}
s_{V} &> 2 G_{11}^2,&
\nonumber
s_{44} &> \frac{2 G_{11}^2}{3},&
\nonumber
G_{11}^2 &> 0,\\
\nonumber
s_{12}&\in\mathbb R,&
\nonumber
\lambda_-^2&>0,&
\lambda_+^2&>0.
\end{align}
To simplify that condition further we define new parameters $p_1$, $p_2$, $p_3$ as follows
\begin{align}
\nonumber
p_1 &= \frac1{\lambda_-^2}\left(2s_{22}(s_V-2G_{11}^2)(1 + 3\lambda_+^2)-\lambda_-^2s_{12}^2\right),\\
p_2 &= N_2/D_2,\quad
p_3 = \frac{3s_{44}}{2 G_{11}^2}-1,\\
\nonumber
-N_2 &=
s_{22}(2G_{11}^2 - s_V)\times\\
\nonumber&\qquad
\times\left(2G_{11}^2(4 + 9\lambda_+^2) + 9\lambda_+^2s_{44} - 4 s_V(1 + 3\lambda_+^2)\right)\\
\nonumber&\quad
+(2G_{11}^2 - s_V)2\lambda_-^2 s_{12}^2\\
\nonumber
D_2 &=
4(G_{11}^2)^2(4 + 9\lambda_+^2)s_{22}
- 9\lambda_+^2s_{22}s_{44}s_V\\
\nonumber&\quad + 2G_{11}^2\left(2\lambda_-^2 s_{12}^2 + 9\lambda_+^2 s_{22} s_{44} - (4 + 9\lambda_+^2)s_{22}s_V\right)
\end{align}
The positivity condition for $\tilde G^{ss}$ then becomes just
\begin{align}
\label{eq:simple_pos_Gss_cond}
\nonumber
p_1&>0,&
p_2&>0,&
p_3&>0,&
G_{11}^2&>0,\\
s_{12}&\in\mathbb R,&
\lambda_-^2&>0,&
\lambda_+^2&>0.
\end{align}

\paragraph*{Positivity of the transverse block.}
Now we consider the remaining transverse block \eqref{eq:remaining_pos_mat}.
It is easy to check that the transformation \eqref{eq:trafo_cand_diag_basis}
does not change this matrix.

We use the solution of \eqref{eq:candidate_condition} to write
\eqref{eq:remaining_pos_mat} in terms of
$p_1$, $p_2$, $p_3$, $G_{11}^2$, $s_{12}$, $\lambda_-^2$, $\lambda_+^2$, $A_{11}^1$ and $A_{33}^1$
and show that the parameters can be chosen to make it positive definite
for every $\lambda_-^2>0$, $\lambda_+^2>0$.

First we notice that \eqref{eq:remaining_pos_mat} has the following structure
\begin{align}
 G^{AB\,kl\,mn} &=
\left(
\begin{array}{ccc}
G_{11}^{AB\,kl\,mn} & G_{12}^{AB\,kl} & 0 \\
G_{12}^{AB\,mn} & G_{22}^{AB} & 0 \\
0 & 0 & G_{33}^{AB\,k\,m}
\end{array}
\right),
\end{align}
where the component tensors can be written in terms of $q$ and $s$:
\begin{subequations}
\label{eq:tensor_struct_rest}
\begin{align}
\label{eq:tensor_struct_rest_a}
\nonumber
G_{11}^{AB\,kl\,mn} &=
  g_{11}^1 q^{AB}q^{kl}q^{mn}
+ 2g_{11}^2 q^{AB}q^{k(m}q^{n)l}\\
\nonumber&\qquad
+ 2g_{11}^3 (q^{A(m}q^{n)B}q^{kl}+q^{A(k}q^{l)B}q^{mn})\\
\nonumber&\qquad
+ 2g_{11}^4 (q^{Ak}q^{B(m}q^{n)l}+q^{Al}q^{B(m}q^{n)k})\\
\nonumber&\qquad
+ 2g_{11}^5 (q^{Am}q^{B(k}q^{l)n}+q^{An}q^{B(k}q^{l)m})\\
\nonumber&\qquad
+ g_{11}^6 q^{AB}(s^ks^l q^{mn}+q^{kl}s^ms^n)\\
\nonumber&\qquad
+ 2g_{11}^7 q^{AB}(s^ks^{(m} q^{n)l}+s^ls^{(m} q^{n)k})\\
\nonumber&\qquad
+ 2g_{11}^8 (q^{A(m}q^{n)B}s^{k}s^{l}+q^{A(k}q^{l)B}s^{m}s^{n})\\
\nonumber&\qquad
+ 2g_{11}^9 (q^{Ak}q^{B(m}s^{n)}s^{l}+q^{Al}q^{B(m}s^{n)}s^{k})\\
\nonumber&\qquad
+ 2g_{11}^{10} (q^{Am}q^{B(k}s^{l)}s^{n}+q^{An}q^{B(k}s^{l)}s^{m})\\
&\qquad
+ g_{11}^{11} q^{AB} s^ks^ls^ms^n,\\
\label{eq:tensor_struct_rest_b}
G_{12}^{AB\,kl} &=
  g_{12}^1 q^{AB} q^{km}
+ 2g_{12}^2 q^{A(k} q^{m)B}
+ g_{12}^3 q^{AB} s^k s^m,\\
\label{eq:tensor_struct_rest_c}
G_{22}^{AB} &= g_{22}^1 q^{AB},\\
\label{eq:tensor_struct_rest_d}
\nonumber
G_{33}^{AB\,k\,m} &=
  g_{33}^1 q^{AB} q^{km}
+ g_{33}^2 q^{Ak} q^{Bm}\\
&\qquad
+ g_{33}^3 q^{Am} q^{kB}
+ g_{33}^4 q^{AB} s^k s^m.
\end{align}
\end{subequations}

When we assume that $G_{22}$ is positive definite then
we can transform this matrix without altering positivity to
\begin{align}
\mbox{diag}(
\bar G_{11}^{AB\,kl\,mn},\,
&G_{22}^{AB},\,G_{33}^{AB\,k\,m}),
\end{align}
where
\begin{align}
\bar G_{11}^{AB\,kl\,mn}:=
G_{11}^{AB\,kl\,mn}-G_{12}^{AC\,kl}(G_{22}^{-1})_{CD}G_{12}^{DB\,mn}.
\end{align}
The structure of $\bar G_{11}^{AB\,kl\,mn}$ is of course
\eqref{eq:tensor_struct_rest_a}, too. Yet, we denote the scalar parameters
$\bar g_{11}^c$ instead of $g_{11}^c$.

\paragraph*{Positivity of $G_{22}^{AB}$.}
Concerning the matrix $G_{22}^{AB}$ we find $g_{22}^1 = \lambda_-^2p_1^2/D$,
with the denominator
\begin{align}
D&=3G_{11}^2\lambda_+^2(p_1+s_{12}^2)p_3(1+p_2)+4p_1p_2(3\lambda_+^2+1).
\end{align}
We see that with the condition \eqref{eq:simple_pos_Gss_cond}
every term in $N$ and $D$ is positive. Hence, $G_{22}^{AB}$ is positive.

\paragraph*{Positivity of $G_{33}^{AB\,k\,m}$.}
If we write the matrix $G_{33}^{AB\,k\,m}$ in an orthonormal basis that
contains $s$ then we find the following eigenvalues
\begin{align}
\{g_{33}^2,\,
g_{33}^1-g_{33}^4,\,
g_{33}^1+g_{33}^4,\,
g_{33}^1+2g_{33}^3+g_{33}^4\}.
\end{align}
Written in terms of
$p_1$, $p_2$, $p_3$, $G_{11}^2$, $s_{12}$, $\lambda_-^2$, $\lambda_+^2$, $A_{11}^1$ and $A_{33}^1$
those eigenvalues have the form
\begin{subequations}
\begin{align}
g_{33}^2 &= z_1(A_{33}^1-z_2)(A_{33}^1-z_3),\\
g_{33}^1-g_{33}^4 &= A_{33}^1-z_2,\\
g_{33}^1+g_{33}^4 &= A_{33}^1+z_3,\\
g_{33}^1+2g_{33}^3+g_{33}^4 &= z_4(A_{33}^1-z_5)(A_{33}^1-z_3).
\end{align}
\end{subequations}
The expressions for $z_1,\ldots,z_5$ are quite complicated functions of the
parameters. But using \eqref{eq:simple_pos_Gss_cond} and computer algebra one can show that
that $z_1<0$, $z_4<0$, $z_2<z_3$ and $z_5<z_3$.
Hence, we can choose $A_{33}^1$ such that $G_{33}^{AB\,k\,m}$ is positive
definite, e.g. $A_{33}^1 = (\max(z_2,z_5)+z_3)/2$.

Since $A_{33}^1$ does not appear in $G_{22}^{AB}$ or $\bar G_{11}^{AB\,kl\,mn}$
we impose conditions that do not restrict the parameter choices in the rest of the
matrix.

\paragraph*{Positivity of $\bar G_{11}^{AB\,kl\,mn}$.}

The hardest part in the positivity analysis is the matrix $\bar G_{11}^{AB\,kl\,mn}$.
Again the first step is to write it in an orthonormal basis that contains $s$. One finds
that the resulting matrix decomposes into two $4\times 4$ and two $2\times 2$ blocks,
where the $4\times 4$ blocks are similar. The $2\times 2$ blocks are
\begin{align}
\label{eq:barG_11_2x2}
\nonumber
&2\left(
\begin{array}{cc}
\bar g_{11}^7 & \bar g_{11}^{10}\\
\bar g_{11}^{10} & \bar g_{11}^7
\end{array}
\right),\\
&2\left(
\begin{array}{cc}
\bar g_{11}^7+\bar g_{11}^{10}+\bar g_{11}^9 & \bar g_{11}^9\\
\bar g_{11}^9 & \bar g_{11}^7+\bar g_{11}^{10}+\bar g_{11}^9
\end{array}
\right).
\end{align}

Looking at the entries in the diagonal of the $4\times 4$ blocks
one finds complicated expressions only at two
positions. Those entries are
\begin{align*}
2 (\bar g_{11}^1 + \bar g_{11}^2 + 2 \bar g_{11}^3 + \bar g_{11}^4 + \bar g_{11}^5)
\quad\mbox{and}\quad
\bar g_{11}^{11}.
\end{align*}

Using a congruence transformation we decompose the $4\times 4$ blocks into $2\times 2$ matrices
such that the complicated expressions are in only one $2\times 2$ matrix. The resulting matrices
are
\begin{subequations}
\label{eq:G11_2x2_decomp}
\begin{align}
\label{eq:G11_2x2_decomp_a}
2&\left(
\begin{array}{cc}
\bar g_{11}^2+\bar g_{11}^4+\bar g_{11}^5 & \bar g_{11}^5-\bar g_{11}^4\\
\bar g_{11}^5-\bar g_{11}^4 & \bar g_{11}^2 \bar g_{11}^4+\bar g_{11}^5
\end{array}
\right),\\
\label{eq:G11_2x2_decomp_b}
&\frac1{\bar g_{11}^2+2\bar g_{11}^4}\left(
\begin{array}{cc}
2 m_{11} & m_{12}\\
m_{12} & m_{22}
\end{array}
\right),
\end{align}
\end{subequations}
where
\begin{align}
\nonumber
m_{11} &=
(\bar g_{11}^2)^2
+ \bar g_{11}^1 (\bar g_{11}^2 + 2 \bar g_{11}^4)
+ \bar g_{11}^2 (2 \bar g_{11}^3 + 3 \bar g_{11}^4 + \bar g_{11}^5)\\
\nonumber&\qquad
- 2 ((\bar g_{11}^3)^2 + 2 \bar g_{11}^3 \bar g_{11}^5 + \bar g_{11}^5 (\bar g_{11}^4 + \bar g_{11}^5)),\\
\nonumber
m_{12} &= \sqrt{2} (2 \bar g_{11}^4 \bar g_{11}^6 - 
   2 (\bar g_{11}^3 + \bar g_{11}^5) \bar g_{11}^8 + \bar g_{11}^2 (\bar g_{11}^6 + \bar g_{11}^8)),\\
m_{22} &= \bar g_{11}^{11} \bar g_{11}^2 + 
 2 \bar g_{11}^{11} \bar g_{11}^4 - 2 (\bar g_{11}^8)^2.
\end{align}

To simplify the form of the eigenvalues of the $2\times 2$ blocks \eqref{eq:barG_11_2x2},
\eqref{eq:G11_2x2_decomp_a} we define another parameter, $p_4$, in terms of $A_{11}^1$:
\begin{align}
\nonumber
p_4 &= \frac{ G_{11}^2 (1 + p_2) }{32 p_1}\left(3\lambda_+^2p_3(p_1+s_{12}^2)+4p_1(1+3\lambda_+^2)\right)\\
&\qquad - A_{11}^1.
\end{align}
We get the following eigenvalues
\begin{align}
\nonumber
\bigg\{&\frac{4\left(3(G_{11}^2)^2\lambda_+^2 + 3G_{11}^2(\lambda_+^2-1)p_4 - 4 p_4^2\right)}{(1+3\lambda_+^2) G_{11}^2},\\
&3 G_{11}^2-2 p_4,
4 p_4,6 p_4-\frac{4 p_4^2}{G_{11}^2},10 p_4-\frac{12 p_4^2}{\lambda_+^2 G_{11}^2}\bigg\}.
\end{align}
Hence, together with \eqref{eq:simple_pos_Gss_cond},
the three matrices \eqref{eq:barG_11_2x2}, \eqref{eq:G11_2x2_decomp_a}
are positive definite if and only if the inequalities
\begin{align}
\label{eq:pos_cond_simple_11_blocks}
\nonumber
p_4 &> 0, &
p_4 &< \frac32 G_{11}^2, &
p_4 &< \frac{5}{6} G_{11}^2\lambda_+^2,\\
\nonumber
p_1 &> 0, &
p_2 &> 0,&
p_3 &> 0,\\
s_{12} &\in \mathbb R,&
\lambda_-^2 &> 0
\end{align}
are satisfied.

Now, to prove positivity of the remaining block \eqref{eq:G11_2x2_decomp_b}
we use the fact that the only condition on $s_{12}$
is $s_{12}\in\mathbb R$. We choose $s_{12}$ such that \eqref{eq:G11_2x2_decomp_b}
becomes diagonal, i.e. we solve the equation $m_{12}=0$ for $s_{12}$.
We get $s_{12}^2=-N_{12}/D_{12}-p_1$ with the numerator and denominator
{
\allowdisplaybreaks
\begin{align}
\label{eq:s_{12}_solution}
\nonumber
N_{12} &= 4(1 + 3\lambda_+^2)p_1p_2\times\\
\nonumber
&\quad\times
\big[3 (G_{11}^2)^3\lambda_+^4 p_3 + 3(G_{11}^2)^2\lambda_+^2((\lambda_+^2-1)p_3-10)p_4\\
\nonumber
&\qquad + 
  4G_{11}^2(9 - \lambda_+^2(p_3-5))p_4^2 - 24 p_4^3\big],\\
\nonumber
D_{12} &= 3\lambda_+^2 p_3
\big[3 (G_{11}^2)^3 \lambda_+^4 p_2 p_3\\
\nonumber
&\quad\qquad
+ 3 (G_{11}^2)^2\lambda_+^2((\lambda_+^2-1) p_2 p_3 -10(1 + p_2)) p_4\\
\nonumber
&\quad\qquad
+ 4 G_{11}^2 ((9 + 5\lambda_+^2)(1 + p_2) - \lambda_+^2 p_2 p_3) p_4^2\\
&\quad\qquad
- 24(1 + p_2)p_4^3\big].
\end{align}
}
The condition $s_{12}\in\mathbb R$ can be written as $s_{12}^2\geq 0$. It
restricts the allowed range of $p_2$ and $p_3$:
{\allowdisplaybreaks
\begin{align}
\nonumber
10p_3 &>\left(\frac{G_{11}^2\lambda_+^2}{5 G_{11}^2\lambda_+^2 - 6p_4} + \frac{5G_{11}^2\lambda_+^2}{3G_{11}^2 - 2p_4} + \frac{G_{11}^2\lambda_+^2}{p_4}\right)^{-1},\\
p_2 &\geq \frac{6 \lambda_+^2 p_3 (5 G_{11}^2\lambda_+^2 - 6p_4)(3 G_{11}^2 - 2p_4)p_4}{(4 + 3\lambda_+^2(4 + p_3)) D_2},\\
\nonumber
D_2&=3(G_{11}^2)^3\lambda_+^4p_3 + 3(G_{11}^2)^2\lambda_+^2\left((\lambda_+^2-1)p_3-10\right)p_4\\
\nonumber
&\qquad + 4G_{11}^2\left(9 - \lambda_+^2(p_3-5)\right)p_4^2 - 24p_4^3,
\end{align}
}
but if these parameters are chosen sufficiently big then the inequalities are satisfied.

If we replace $s_{12}$ everywhere using \eqref{eq:s_{12}_solution} then
\eqref{eq:G11_2x2_decomp_b} has only one simple eigenvalue:
\begin{align}
\frac{2(5G_{11}^2\lambda_+^2 - 6p_4)(3G_{11}^2 - 2p_4)p_4}{3(G_{11}^2)^2\lambda_+^2 + 3G_{11}^2(\lambda_+^2-1)p_4 - 4p_4^2}.
\end{align}
One can check that numerator and denominator are positive if \eqref{eq:pos_cond_simple_11_blocks}
is satisfied.

Hence, indeed for every $\lambda_-^2>0$, $\lambda_+^2>0$
the parameters can be chosen such that
$\bar G_{11}^{AB\,kl\,mn}$ is positive definite.

The calculations presented here were performed using the computer algebra
system {\tt mathematica} \cite{Mathematica} with the package {\tt xTensor} by
Jos\'e-Mar\'ia Mart\'in-Garc\'ia \cite{xAct}.


\bibliography{bibliography}{}
\bibliographystyle{unsrt}


\end{document}